\documentclass{tlp}

\usepackage[utf8]{inputenc}
\usepackage[english]{babel}

\usepackage{latexsym}
\usepackage{etoolbox}
\usepackage{filecontents}
\usepackage{amsmath}
\usepackage{amssymb}
\usepackage{mathptmx}
\usepackage{marvosym}

\usepackage{microtype}
\usepackage{verbatim}
\usepackage{booktabs}
\usepackage{natbib}

\makeatletter
\let\underscore\_
\def\_{\underscore\@ifnextchar f{\kern1pt}{\kern1pt}}
\def\ps@titlepage{
  \def\@oddhead{\affilsize\textit{Accepted for publication in Theory and Practice of Logic Programming}\hfil\llap{\thepage}}%
  \let\@oddfoot\@empty
  \let\@evenfoot\@oddfoot
}
\makeatother
\submitted{31 December 2014}
\revised{21 May 2015}
\accepted{23 September 2015}
\journal{Accepted for publication in Theory and Practice of Logic Programming}

\usepackage{xspace}
\usepackage{xcolor}

\usepackage{algorithmicx}
\usepackage{algorithm}
\usepackage{algpseudocode}
\usepackage{lineno}

\newcommand{\Ifline}[2]{\textbf{if} #1 \textbf{then} #2}

\algloopdefx{Switch}[1]{\textbf{switch} (#1)}%

\algblockdefx{Case}{EndCase}[1]{\textbf{case} #1\textbf{:}}{Break}

\algcblockdefx{Case}{Case}{EndCase}[1]{\textbf{case} #1\textbf{:}}{\hspace{-0.5cm}\textbf{end switch}}

\algcblockdefx{Case}{Otherwise}{EndOther}{\textbf{otherwise:}}{\hspace{-0.5cm}\textbf{end switch}}

\AtBeginEnvironment{enumerate}{\ifintheorem\advance\leftmargin\leftmargini\fi}
\AtBeginEnvironment{itemize}{\ifintheorem\advance\leftmargin\leftmargini\fi}

\usepackage[inline]{enumitem}

\undef{\openbox}
\undef{\proof}
\undef{\endproof}

\usepackage{amsthm}
\usepackage{thmtools,thm-restate}
\declaretheoremstyle[
  spaceabove=5pt, spacebelow=5pt,
  headfont=\normalfont\itshape,
  notefont=\mdseries, notebraces={(}{)},
bodyfont=\normalfont,
  headpunct=\\,
  postheadhook=\intheoremtrue,
]{mystyle}

\declaretheorem[style=mystyle,heading=Theorem]{theorem}

\declaretheorem[style=mystyle,heading=Lemma,sibling=theorem]{lem}

\declaretheorem[style=mystyle,heading=Definition]{definition}
\declaretheorem[style=mystyle,heading=Example]{example}

\undef{\proof}
\undef{\endproof}
\declaretheorem[style=mystyle,qed=\qedsymbol,numbered=no]{proof}

\newcommand{\non}{\mathnormal{\sim}}
\newcommand{\To}{\Rightarrow}

\newcommand{\ds}{\displaystyle}
\newcommand{\oseq}[2][n]{#2_{1}\odot\cdots\odot #2_{#1}}
\newcommand{\opseq}[4][\otimes]{#2_{#3}#1\cdots#1#2_{#4}}

\newcommand{\set}[2][\relax]{\ensuremath{#1\{#2#1\}}}
\newcommand{\Set}[2][\relax]{%
  \ensuremath{
    \ifx#1\left
      #1\{#2\right\}
    \else\ifx#1\right
      \left\{#2#1\}
      \else #1\{#2#1\}\fi\fi}%
}
\RequirePackage{ifthen}
\renewcommand{\Set}[2][\relax]{%
  \ifthenelse{\equal{#1}{auto}}{\left\{#2\right\}}{#1\{#2#1\}}%
}
\newcommand{\seq}[2][n]{\ensuremath{#2_1,\dots,#2_{#1}}}

\newcommand{\LAB}{\ensuremath{\mathrm{Lab}}\xspace}

\newcommand{\FACTS}{\ensuremath{\mathrm{F}}\xspace}

\newcommand{\OUTT}{\ensuremath{\mathsf{OUT}}\xspace}
\newcommand{\BELL}{\ensuremath{\mathsf{BEL}}\xspace}
\newcommand{\OBLL}{\ensuremath{\mathsf{OBL}}\xspace}

\newcommand{\OBL}{\ensuremath{\mathsf{O}}\xspace}
\newcommand{\BEL}{\ensuremath{\mathsf{B}}\xspace}
\newcommand{\GOAL}{\ensuremath{\mathsf{G}}\xspace}
\newcommand{\INT}{\ensuremath{\mathsf{I}}\xspace}
\newcommand{\INTS}{\ensuremath{\mathsf{SI}}\xspace}
\newcommand{\DES}{\ensuremath{\mathsf{D}}\xspace}
\newcommand{\OUT}{\ensuremath{\mathsf{U}}\xspace}

\let\Bel\BEL

\newcommand{\Convert}{\ensuremath{\mathrm{Convert}}\xspace}
\newcommand{\Conflict}{\ensuremath{\mathrm{Conflict}}\xspace}

\newcommand{\Aconf}[2]{\mathrm{Conflict}(#1,#2)}
\newcommand{\Aconv}[2]{\mathrm{Convert}(#1,#2)}

\newcommand{\infd}{\ensuremath{\mathit{infd}}\xspace}

\newcommand{\PROP}{\ensuremath{\mathrm{PROP}}\xspace}

\newcommand{\MOD}{\ensuremath{\mathrm{MOD}}\xspace}
\newcommand{\LIT}{\ensuremath{\mathrm{Lit}}\xspace}
\newcommand{\MODLIT}{\ensuremath{\mathrm{ModLit}}\xspace}

\newcommand{\LRangle}{\langle\ \rangle}
\newcommand{\LRset}[2][\relax]{\ensuremath{#1\langle#2#1\rangle}}

\newcommand{\psl}{\phantom{{}=\{}} 

\begin{document}
	
\pagestyle{plain}

\title{The Rationale behind the Concept of Goal}
\author{Guido Governatori$^{1}$, Francesco Olivieri$^{2}$, Simone Scannapieco$^{2}$,\\
  \normalfont\normalsize Antonino Rotolo$^3$ and Matteo Cristani$^2$\\
  \\
  $^1$ NICTA, Australia\\
  $^2$ Department of Computer Science, Verona, Italy\\
  $^3$ CIRSFID, University of Bologna, Bologna, Italy}

\maketitle
\begin{abstract}

The paper proposes a fresh look at the concept of goal and advances that
motivational attitudes like desire, goal and intention are just facets of the
broader notion of (acceptable) outcome. We propose to encode the preferences
of an agent as sequences of ``alternative acceptable outcomes''. We then
study how the agent's beliefs and norms can be used to filter the mental
attitudes out of the sequences of alternative acceptable outcomes. Finally, we 
formalise such intuitions in a novel Modal Defeasible Logic and we prove that
the resulting formalisation is computationally feasible.

\end{abstract}


\section{Introduction and motivation}
\label{sec:Intro}

\newtheorem{challenge}{Desiderata}

The core problem we address in this paper is how to formally describe a system operating in
an environment, with some objectives to achieve, and trying not to violate the
norms governing the domain in which the system operates.

To model such systems, we have to specify three types of information: (i) the
environment where the system is embedded, i.e., how the system perceives the
world, (ii) the norms regulating the application domain, and (iii) the
system's internal constraints and objectives.

A successful abstraction to represent a system operating in an environment
where the system itself must exhibit some kind of autonomy is that of BDI
(Belief, Desire, Intention) architecture \citep{DBLP:conf/kr/RaoG91} inspired
by the work of \cite{bratman1} on cognitive agents. In the BDI architecture,
desires and intentions model the agent's mental attitudes and are meant to
capture the objectives, whereas beliefs describe the environment. More
precisely, the notions of belief, desire and intention represent
respectively the informational, motivational and deliberative states of an
agent \citep{DBLP:conf/ecaiw/WooldridgeJ94}.

Over the years, several frameworks, either providing extensions of BDI or
inspired by it, were given with the aim of extending models for cognitive
agents to also cover normative aspects (see, among others,
\citep{broersen2002goal, thomason, jaamas:08}). (This is a way of developing
normative agent systems, where norms are meant to ensure global properties
for them \citep{andrighetto_et_al:DR:2012:3535}.) In such extensions, the
agent behaviour is determined by the interplay of the cognitive component and
the normative one (such as obligations). In this way, it is possible to
represent how much an agent is willing to invest to reach some outcomes based
on the states of the world (what we call beliefs) and norms. Indeed, beliefs
and norms are of the utmost importance in the decision process of the agent.
If the agent does not take beliefs into account, then she will not be able to
plan what she wants to achieve, and her planning process would be a mere
wishful thinking. On the other hand, if the agent does not respect the norms
governing the environment she acts in, then she may incur sanctions from
other agents \citep{bratman1}.



The BDI approach is based on the following assumptions about the motivational
and deliberative components. The agent typically defines \emph{a priori} her
desires and intentions, and only after this is done the system verifies their
mutual consistency by using additional axioms. Such entities are therefore
not interrelated with one another since ``the notion of intention [\dots] has
equal status with the notions of belief and desire, and cannot be reduced to
these concepts'' \citep{DBLP:conf/kr/RaoG91}. Moreover, the agent may
consequently have intentions which are contradictory with her beliefs and
this may be verified only \emph{a posteriori}. Therefore, one of the main
conceptual deficiencies of the BDI paradigm (and generally of almost all
classical approaches to model rational agents) is that the deliberation
process is bound to these mental attitudes which are independent and fixed
\emph{a priori}. Here, with the term independent, we mean that none of them
is fully definable in terms of the others.

Approaches like the BOID (Belief, Obligation, Intention, Desire) architecture \citep{broersen2002goal} and \cite{jaamas:08}'s system improve
previous frameworks, for instance, by structurally solving conflicts between
beliefs and intentions (the former being always stronger than any conflicting
intention), while mental attitudes and obligations are just meant to define
which kinds of agent (social, realistic, selfish, and so on) are admissible. 


Unlike the BDI perspective, this paper aims at proposing a
fresh conceptual and logical analysis of the motivational and deliberative
components within a unified perspective.

\paragraph{Desideratum 1: A unified framework for agents' motivational and
deliberative components. }{Goals, desires, and intentions are \emph{different
facets} of the \emph{same phenomenon}, all of them being goal-like attitudes.
This reduction into a unified perspective is done by resorting to the basic
notion of \emph{outcome}, which is simply something (typically, a state of
affairs) that an agent expects to achieve or that can possibly occur.}

\smallskip

Even when considering the vast literature on goals of the past decade,
most of the authors studied the content of a goal (e.g., \emph{achievement}
or \emph{maintenance} goals) and conditions under which a goal has to be
either pursued, or dropped. This kind of (\emph{a posteriori}) analysis
results orthogonal to the one proposed hereafter, since we want to develop a
framework that computes the agent's mental attitudes by combining her beliefs
and the norms with her desires.

As we shall argue, an advantage of the proposed analysis is that it allows
agents to compute different degrees of motivational attitudes, as well as
different degrees of commitment that take into account other, external,
factors, such as \emph{beliefs} and \emph{norms}.

\paragraph{Desideratum 2: Agents' motivations emerge from preference
orderings among outcomes.}{The motivational and deliberative components of
agents are generated from preference orderings among outcomes. As done in
other research areas (e.g., rational choice theory), we move with the idea
that agents have preferences and choose the actions to bring about according
to such preferences. Preferences involve outcomes and are explicitly
represented in the syntax of the language for reasoning about agents, thus
following the logical paradigm initially proposed in
\citep{BrewkaBB04,ajl:ctd}.}

\smallskip

The combination of an agent's mental attitudes with the factuality of the
world defines her deliberative process, i.e., the objectives she decides to
pursue. The agent may give up some of them to comply with the norms, if
required. Indeed, many contexts may prevent the agent from achieving all of
her objectives; the agent must then understand which objectives are mutually
compatible with each other and choose which ones to attain the least of in
given situations by ranking them in a preference ordering.

The approach we are going to formalise can be summarised as follows. We distinguish
three phases an agent must pass through to bring about certain states of
affairs: (i) The agent first needs to understand the environment she acts in;
(ii) The agent deploys such information to deliberate which objectives to
pursue; and (iii) The agent lastly decides how to act to reach them.

In the first phase, the agent gives a formal declarative description of the
environment (in our case, a rule-based formalism). Rules allow the agent to
represent relationships between pre-conditions and actions, actions and their
effects (post-conditions), relationships among actions, which conditions
trigger new obligations to come in force, and in which contexts the agent is
allowed to pursue new objectives.

In the second phase, the agent combines the formal description with an input
describing a particular state of affairs of the environment, and she
determines which norms are actually in force along with which objectives she
decides to commit to (by understanding which ones are attainable) and to
which degree. The agent's decision is based on logical derivations.

Since the agent's knowledge is represented by rules, during the third
and last phase, the agent combines and exploits all such information obtained 
from the conclusions derived in the second phase to
select which activities to carry out in order to achieve the objectives. (It
is relevant to notice that a derivation can be understood as a virtual
simulation of the various activities involved.)

While different schemas for generating and filtering agents' outcomes are
possible, the three phases described above suggest to
adopt the following principles:
\begin{itemize}
\item When an agent faces alternative outcomes in a given context, these
  outcomes are ranked in preference orderings;
\item Mental attitudes are obtained from a single type of rule
  (\emph{outcome rule}) whose conclusions express the above mentioned
  preference orderings among outcomes;
\item Beliefs prevail over conflicting motivational attitudes, thus avoiding
  various cases of wishful thinking \citep{thomason,broersen2002goal};
\item Norms and obligations are used to filter social motivational states
(\emph{social intentions}) and compliant agents
\citep{broersen2002goal,jaamas:08};
\item Goal-like attitudes can also be derived via a \emph{conversion}
mechanism using other mental states, such as beliefs \citep{jaamas:08}. For
example, believing that Madrid is in Spain may imply that the goal to go to
Madrid implies the goal to go to Spain.
\end{itemize}

Our effort is finally motivated by computational concerns. The logic for
agents' desires, goals, and intentions is expected to be computationally
efficient. In particular, we shall prove that computing agents' motivational
and deliberative components in the proposed unified framework has linear
complexity.

\section{The intuition underneath the framework} 
\label{sec:Intuition}


When a cognitive agent deliberates about what her outcomes are in a particular
situation, she selects a set of \emph{preferred} outcomes among a larger
set, where each specific outcome has various alternatives. It is natural to
rank such alternatives in a preference ordering, from the most preferred
choice to the least objective she deems acceptable.

Consider, for instance, the following scenario. Alice is thinking what to do
on Saturday afternoon. She has three alternatives: (i) she can visit John;
(ii) she can visit her parents who live close to John's place; or (iii) she
can watch a movie at home. The alternative she likes the most is visiting
John, while watching a movie is the least preferred. If John is not at home,
there is no point for Alice to visit him. In this case, paying a visit to her
parents becomes the ``next best'' option. Also, if visiting her parents is
not possible, she settles for the last choice, that of staying home and
watching a movie.

Alice also knows that if John is away, the alternative of going to his place
makes no sense. Suppose that Alice knows that John is actually away for the
weekend. Since the most preferred option is no longer available, she decides
to opt for the now best option, namely visiting her parents.

To represent the scenario above, we need to capture the preferences about her
alternatives, and her beliefs about the world. To model preferences
among several options, we build a sequence of alternatives $\seq{A}$ that are
preferred when the previous choices are no longer feasible. Normally, each set of
alternatives is the result of a specific context $C$ determining 
under which conditions (premises) such a sequence of alternatives $\seq{A}$
is considered.

Accordingly, we can represent Alice's alternatives with the notation
\begin{gather*}
  \textbf{If}\quad \mathit{saturday}\quad \textbf{then}\quad \mathit{visit\_John}, \ \mathit{visit\_parents}, \ \mathit{watch\_movie}.  
\end{gather*}
This intuition resembles the notion of contrary-to-duty obligations
presented by \cite{ajl:ctd}, where a norm is represented by an \emph{obligation rule}
of the type
\begin{gather*}
  r_1: \mathit{drive\_car} \Rightarrow_{\OBLL} \neg \mathit{damage}\odot
  \mathit{compensate} \odot \mathit{foreclosure}
\end{gather*}
where ``$\To_{\OBLL}$'' denotes that the conclusion of the rule will be treated as an
obligation, and the symbol ``$\odot$'' replaces the symbol ``,'' to separate
the alternatives. In this case, each element of the chain is the reparative obligation that shall come in
force in case the immediate predecessor in the chain has been violated. Thus, the
meaning of rule $r_1$ is that, if an agent drives a car, then she has the
obligation not to cause any damage to others; if this happens, she is obliged
to compensate; if she fails to compensate, there is an obligation of
foreclosure.

Following this perspective, we shall now represent the previous scenario with a
rule introducing the outcome mode, that is an \emph{outcome rule}:
\begin{gather*}
  r_2: \mathit{saturday} \Rightarrow_{\OUTT} \mathit{visit\_John}\odot
  \mathit{visit\_parents} \odot \mathit{watch\_movie}.
\end{gather*}
In both examples, the sequences express a preference
ordering among alternatives. Accordingly, $\mathit{watch\_movie}$ and
$\mathit{foreclosure}$ are the last (and least) acceptable situations.

To model beliefs, we use \emph{belief rules}, like 
\begin{gather*}
r_3: \mathit{John\_away} \Rightarrow_{\BELL} \neg    
    \mathit{visit\_John}
\end{gather*}
meaning that if Alice has the belief that John is not home, then she adds to
her beliefs that it is not possible to visit him.

In the rest of the section, we shall illustrate the principles and intuitions
relating sequences of alternatives (that is, outcome rules), beliefs, obligations, and
how to use them to characterise different types of goal-like attitudes and
degrees of commitment to outcomes: \emph{desires}, \emph{goals}, \emph{intentions}, and \emph{social
intentions}.

\paragraph{Desires as acceptable outcomes.} Suppose that an agent is equipped with the following outcome rules expressing two preference orderings:
\begin{gather*}
  r: \seq{a} \Rightarrow_{\OUTT} \oseq[m]{b} \qquad \qquad 
  s: \seq{a'} \Rightarrow_{\OUTT} \oseq[k]{b'}
\end{gather*}
%
and that the situations described by $\seq{a}$ and $\seq{a'}$ are mutually
compatible but $b_{1}$ and $b'_{1}$ are not, namely $b_{1}=\neg b'_{1}$. In
this case $\seq[m]{b},\seq[k]{b'}$ are all \emph{acceptable outcomes},
including the incompatible outcomes $b_{1}$ and $b'_{1}$. 

\emph{Desires are acceptable outcomes}, independently of whether
they are compatible with other expected or acceptable outcomes.
%
Let us contextualise the previous example to better explain the notion of desire by considering the following setting.
\begin{example}\label{ex:FriendsWillBeFriends}\ \vspace*{-\baselineskip}
\[
	F = \{\mathit{saturday},\, \mathit{John\_sick} \} \quad
	R = \{r_2,\, r_{4}: \mathit{John\_sick} \To_{\OUTT} \neg \mathit{visit\_John}\odot \mathit{short\_visit} \}.
\]
\end{example}
The meaning of $r_{4}$ is that Alice would not visit John if he is sick, but
if she does so, then the visit must be short.

Being the premises of $r_{2}$ and of $r_{4}$ the case, then both rules are
activated, and the agent has both $\mathit{visit\_John}$ and its opposite as
acceptable outcomes. Eventually, she needs to make up her mind. Notice that
if a rule prevails over the other, then the elements of the weaker rule with
an incompatible counterpart in the stronger rule are \emph{not} considered
desires. Suppose that Alice has not visited John for a long time and she has
recently placed a visit to her parents. Then, she prefers to see John instead
of her parents despite John being sick. In this setting, $r_{2}$ prevails
over $r_{4}$ ($r_{2} > r_{4}$ in notation). Given that she explicitly prefers
$r_{2}$ to $r_{4}$, her desire is to visit John ($\mathit{visit\_John}$) and
it would be irrational to conclude that she also has the opposite desire
(i.e., $\neg \mathit{visit\_John}$).

\paragraph{Goals as preferred outcomes.} We consider a \emph{goal as the
preferred desire in a chain}.

For rule $r$ alone the preferred outcome is $b_{1}$, and for rule $s$ alone
it is $b'_{1}$. But if both rules are applicable, then a state where both
$b_{1}$ and $b'_{1}$ hold is not possible: the agent would not be rational if
she considers both $b_{1}$ and $\neg b_{1}$ as her preferred outcomes. Therefore,
the agent has to decide whether she prefers a state where $b_{1}$ holds
to a state where $b'_{1}$ (i.e., $\neg b_{1}$) does (or the other
way around). If the agent cannot make up her mind, i.e., she has no way to
decide which is the most suitable option for her, then neither the chain of
$r$ nor that of $s$ can produce preferred outcomes.

Consider now the scenario where the
agent establishes that the second rule overrides the first one ($s>r$).
Accordingly, the preferred outcome is $b'_{1}$ for the chain of outcomes
defined by $s$, and $b_{2}$ is the preferred outcome of $r$. $b_{2}$ is the
second best alternative according to rule $r$: in fact $b_{1}$ has been
discarded as an acceptable outcome given that $s$ prevails over $r$.

In the situation described by Example~\ref{ex:FriendsWillBeFriends}, $\mathit{visit\_John}$ is the goal according to $r_{2}$, while
$\mathit{short\_visit}$ is the goal for $r_{4}$.

\paragraph{Two degrees of commitment: intentions and social intentions.}
The next issue is to clarify which are the acceptable outcomes for an agent to
commit to. Naturally, if the agent values some outcomes more than others, she
should strive for the best, in other words, for the most preferred outcomes (goals).

We first consider the case where only rule $r$ applies. Here, the agent
should commit to the outcome she values the most, that is $b_{1}$. But
what if the agent \emph{believes} that $b_{1}$ cannot be achieved in the
environment where she is currently situated in, or she knows that $\neg
b_{1}$ holds? Committing to $b_{1}$ would result in a waste of the agent's
resources; rationally, she should target the next best outcome $b_{2}$.
Accordingly, the agent derives $b_{2}$ as her \emph{intention}. \emph{An
intention is an acceptable outcome which does not conflict with the beliefs
describing the environment.}

Suppose now that $b_{2}$ is \emph{forbidden}, and that the agent is social (a social agent is an agent not knowingly committing to anything that is
forbidden \citep{jaamas:08}). Once again, the agent has to
lower her expectation and settle for $b_{3}$, which is one of her
\emph{social intentions}. \emph{A social intention is an intention which does not violate any norm.}

To complete the analysis, consider the situation where both rules $r$ and $s$
apply and, again, the agent prefers $s$ to $r$. As we have seen before, $\neg
b_{1}$ ($b'_{1}$) and $b_{2}$ are the preferred outcomes based on the
preference of the agent over the two rules. This time we assume that the
agent knows she cannot achieve $\neg b_{1}$ (or equivalently, $b_{1}$ holds).
If the agent is rational, she cannot commit to $\neg b_{1}$. Consequently,
the best option for her is to commit to $b'_{2}$ and $b_{1}$ (both regarded
as intentions and social intentions), where she is guaranteed to be
successful.

This scenario reveals a key concept: there are situations where the agent's
best choice is to commit herself to some outcomes that are not her
preferred ones (or even to a choice that she would consider not 
acceptable based only on her preferences) but such that they influence her 
decision process, given that they represent relevant external factors (either 
her beliefs or the norms that apply to her situation).


\begin{example}\label{ex:Jsick}\ \vspace*{-\baselineskip}
\[
	F  = \{\mathit{saturday},\, \mathit{John\_away},\, \mathit{John\_sick} \}\quad
	R  = \{r_2,\, r_3,\, r_{4}\}\quad
	>\  = \{(r_{2}, r_{4}) \}.
\]
Today John is in rehab at the hospital. Even if Alice has the desire as well as the goal to visit John, the facts of the situation lead her to form the intention to visit her parents.

Consider now the following theory
\begin{align*}
	F & =\{\mathit{saturday},\, \mathit{John\_home\_confined},\, \mathit{third\_week} \} \\
	R & = \{r_2,\, r_3,\, r_{4}, r_{5}: \mathit{John\_home\_confined},\, \mathit{third\_week} \To_{\OBLL} \neg \mathit{visit\_John}\}\\
	> & =  \{(r_{2}, r_{4}) \}.
\end{align*}
Unfortunately, John has a stream of bad luck. Now, he is not debilitated but
has been home convicted for a minor crime. The law of his country states that
during the first two months of his home conviction, no visits to him are
allowed. This time, even if Alice knows that John is at home, norms forbid Alice to visit him. Again, Alice opts to visit her parents.
\end{example}

\let\boxbox\Box
\def\myX{X}
\let\Box\myX
\let\boxblacksquare\blacksquare
\def\myY{Y}
\let\blacksquare\myY
\let\boxdiamond\Diamond
\def\myT{T}
\let\Diamond\myT

\section{Logic} 
\label{sec:logic}

Defeasible Logic (DL) \citep{tocl} is a simple, flexible, and efficient rule
based non-monotonic formalism. Its strength lies in its constructive proof
theory, which has an argumentation-like structure, and it allows us to draw
meaningful conclusions from (potentially) conflicting and incomplete knowledge
bases. Being non-monotonic means that more accurate conclusions can be obtained
when more pieces of information are given (where some previously derived
conclusions no longer follow from the knowledge base). 

The framework provided by the proof theory accounts for the possibility of
extensions of the logic, in particular extensions with modal operators.
Several of such extensions have been proposed, which then resulted in
successful applications in the area of normative reasoning \citep{coala},
modelling agents
\citep{jaamas:08,Kravari-ijcai,DBLP:journals/igpl/GovernatoriPRS09}, and
business process compliance \citep{Governatori_Sadiq_2008}. A model theoretic
possible world semantics for modal Defeasible Logic has been proposed in
\citep{deon2012}. In addition, efficient implementations of the logic
(including the modal variants), able to handle very large knowledge bases,
have been advanced in \citep{spindle,drdevice,Antoniou-ECAI2012}.



%
\begin{definition}[Language]
Let \PROP be a set of propositional atoms, and $\MOD = \set{\BEL, \OBL, \DES,
\GOAL, \INT, \INTS}$ the set of modal operators, whose reading is $\BEL$ for
\emph{belief}, $\OBL$ for \emph{obligation}, $\DES$ for \emph{desire}, $\GOAL$
for \emph{goal}, $\INT$ for \emph{intention} and $\INTS$ for \emph{social
intention}. Let \LAB be a set of arbitrary labels. The set $\LIT=\PROP\cup \{
\neg p | p\in\PROP\}$ denotes the set of \emph{literals}. The
\emph{complement} of a literal $q$ is denoted by $\non q$; if $q$ is a
positive literal $p$, then $\non q$ is $\neg p$, and if $q$ is a negative
literal $\neg p$ then $\non q$ is $p$. The set of \emph{modal literals} is
$\MODLIT=\set{\Box l, \neg \Box l | l\in \LIT, \Box\in \set{\OBL, \DES, \GOAL,
\INT, \INTS}}$. We assume that modal operator ``$\Box$'' for belief \BEL is
the empty modal operator. Accordingly, a modal literal $\BEL l$ is equivalent
to literal $l$; the complement of $\BEL \non l$ and $\neg \BEL l$ is $l$. 
\end{definition}

\begin{definition}[Defeasible Theory]\label{def:theory}
A \emph{defeasible theory} $D$ is a structure $(F,R,>)$, where
\begin{enumerate*}[label=(\arabic*)]
\item $F \subseteq \LIT \cup \MODLIT$ is a set of \emph{facts} or indisputable statements;
\item $R$ contains three sets of \emph{rules}: for beliefs, obligations, and outcomes;
\item ${>} \subseteq R \times R$ is a binary \emph{superiority relation} to determine the
relative strength of (possibly) conflicting rules. We use the infix notation $r>s$ to mean that $(r,s)\in >$.
\end{enumerate*}
A theory is \emph{finite} if the set of facts and rules are so.
\end{definition}

\emph{Belief rules} are used to relate the factual knowledge of an agent,
that is to say, her vision of the environment she is situated in. Belief
rules define the relationships between states of the world; as such,
provability for beliefs does not generate modal literals.

\emph{Obligation rules} determine when and which obligations are in force. The
conclusions generated by obligation rules take the $\OBL$ modality. 

Finally,
\emph{outcome rules} establish the possible outcomes of an agent depending on
the particular context. Apart from obligation rules, outcome rules are used to
derive conclusions for all modes representing goal-like attitudes: desires,
goals, intentions, and social intentions.

Following ideas given in \citep{ajl:ctd}, rules can gain more expressiveness
when a \emph{preference operator} $\odot$ is adopted. An expression like $a \odot b$ means that if $a$ is possible,
then $a$ is the first choice, and $b$ is the second one; if $\neg a$ holds,
then the first choice is not attainable and $b$ is the actual choice. This
operator is used to build chains of preferences, called
\emph{$\odot$-expressions}. 
The formation rules for $\odot$-expressions
are:
\begin{enumerate}
\item every literal is an $\odot$-expression;
\item if $A$ is an $\odot$-expression and $b$ is a literal then $A\odot b$ is an $\odot$-expression. 
\end{enumerate}
%
In addition, we stipulate that $\odot$ obeys the following properties: 
\begin{enumerate}
\item $a\odot (b \odot c) = (a \odot b) \odot c$
(associativity);
\item $\bigodot_{i=1}^{n} a_i = (\bigodot_{i=1}^{k-1}a_i) \odot
(\bigodot_{i=k+1}^{n\phantom{k}}a_i)$ where there exists $j$ such that $a_j= a_k$
and $j<k$ (duplication and contraction on the right).
\end{enumerate}
Typically, $\odot$-expressions are given by the agent
designer, or obtained through \emph{construction rules} based on the
particular logic \citep{ajl:ctd}.

In the present paper, we use the classical definition of \emph{defeasible rule} in
DL \citep{tocl}, while \emph{strict rules} and \emph{defeaters} are
omitted\footnote{The restriction does not result in any loss of generality:
(i) the superiority relation does not play any role in proving definite
conclusions, and (ii) for defeasible conclusions \cite{tocl} prove
that it is always possible to remove strict rules from the superiority
relation and defeaters from the theory to obtain an equivalent theory
without defeaters and where the strict rules are not involved in the
superiority relation.}. 
\begin{definition}[Defeasible rule]
A \emph{defeasible rule} is an expression
$r: A(r) \Rightarrow_{\Box} C(r)$, where
\begin{enumerate*}[label=(\arabic*)]
\item $r \in \LAB$ is the name of the rule;
\item $A(r)=\set{\seq{a}}$, the \emph{antecedent} (or \emph{body}) of the rule, is the set of
the premises of the rule. Each $a_i$ is either in $\LIT$ or in $\MODLIT$;
\item $\Box \in \set{\BEL, \OBL, \OUT}$ represents the \emph{mode} of the
rule: $\To_{\BEL}$, $\To_{\OBL}$, $\To_{\OUT}$ denote respectively rules for
beliefs, obligations, and outcomes. From now on, we omit the subscript \BEL in
rules for beliefs, i.e., $\Rightarrow$ is used as a shortcut for
$\Rightarrow_{\BEL}$;

\item $C(r)$ is the \emph{consequent} (or \emph{head}) of the rule, which is a
single literal if $\Box=\BEL$, and an $\odot$-expression otherwise\footnote{It
is worth noting that modal literals can occur only in the antecedent of rules:
the reason is that the rules are used to derive modal conclusions and we do
not conceptually need to iterate modalities. The motivation of a single
literal as a consequent for belief rules is dictated by the intended reading
of the belief rules, where these rules are used to describe the environment.}.
\end{enumerate*}
\end{definition}

A defeasible rule is a rule that can be defeated by contrary evidence. The
underlying idea is that if we know that the premises of the rule are the
case, then we may conclude that the conclusion holds, unless there is
evidence proving otherwise. Defeasible rules in our framework introduce modal
literals; for instance, if we have rule $r: A(r) \To_{\OBL} c$ and the
premises denoted by $A(r)$ are the case, then $r$ can be used to prove
$\OBL$c.

We use the following abbreviations on sets of rules: $R^{\Box}$
($R^{\Box}[q]$) denotes all rules of mode $\Box$ (with consequent $q$), and
$R[q]$ denotes the set $\bigcup_{\Box \in \set{\BEL, \OBL, \OUT}}
R^{\Box}[q]$. With $R[q,i]$ we denote the set of rules whose head is
$\odot_{j=1}^{n}c_{j}$ and $c_{i}=q$, with $1\leq i\leq n$.



Notice that labelling the rules of DL produces nothing more but a
simple treatment of the modalities, thus two interaction strategies between
modal operators are analysed: \emph{rule conversion} and
\emph{conflict resolution} \citep{jaamas:08}.


In the remainder, we shall define a completely new inference machinery that
takes this into account by adding preferences and dealing with a larger set
of modalised conclusions, which are not necessarily obtained from the
corresponding rules but also by using other rule types. For instance, we
argued in Section \ref{sec:Intuition} that a goal can be viewed as a
preferred outcome and so the fact that a certain goal $\GOAL p$ is derived
depends on whether we can obtain $p$ as a preferred outcome by using a rule
for $\OUT$.



\paragraph{Rule conversion.}{It is sometimes meaningful to use rules for a
modality $\Box$ as if they were for another modality $\blacksquare$, i.e., to convert one type of conclusion into a different one. 

Formally, we define an asymmetric binary relation $\Convert \subseteq \MOD
\times \MOD$ such that $\Aconv{\Box}{\blacksquare}$ means ``a rule of mode
$\Box$ can be used also to produce conclusions of mode $\blacksquare$''. This
intuitively corresponds to the following inference schema:
\[
\frac{\ds \seq{\blacksquare a} \quad \seq{a} \To_{\Box} b}
{\ds \blacksquare b} \mbox{\; $\Convert(\Box,\blacksquare)$}. 
\]

In our framework obligations and goal-like attitudes cannot change what the
agent believes or how she perceives the world, we thus consider only
conversion from beliefs to the other modes (i.e., $\Convert(\BEL, \Box)$ with
$\Box \in \MOD\setminus \set{\BEL}$). Accordingly, we enrich the notation with
$R^{\BEL,\Box}$ for the set of belief rules that can be used for a conversion
to mode $\Box \in \MOD \setminus \set{\BEL}$. The antecedent of all such
rules is not empty, and does not contain any modal literal.

\begin{example}\ \vspace*{-\baselineskip}
%
\[
	F = \{\mathit{saturday}\}\quad
	R = \{r_2,\, r_{6}: \mathit{visit\_John} \To \mathit{chocolate\_box} \}
\]
where we stipulate that $\Aconv{\BEL}{\DES}$ holds.

Alice desires to visit John. John is a passionate of chocolate and, usually,
when Alice goes to meet him at his place, she brings him a box of chocolate.
Thus, we may state that her desire of visiting John implies the desire to
bring him a box of chocolate. This is the case since we can use rule $r_{6}$
to convert beliefs into desires.
\end{example}}

\paragraph{Conflict-detection/resolution.}{It is crucial to identify criteria for detecting and solving conflicts between
different modalities.
We define an asymmetric
binary relation $\Conflict \subseteq \MOD \times \MOD$ such that
$\Aconf{\Box}{\blacksquare}$ means ``modes $\Box$ and $\blacksquare$ are in conflict and mode $\Box$ prevails
over $\blacksquare$''. In our framework, we consider conflicts between (i) beliefs and intentions, (ii) beliefs and social intentions, and (iii) obligations and social intentions. In other
words, the agents are characterised by:
\begin{itemize}
  \item $\Conflict(\BEL,\INT)$, $\Conflict(\BEL,\INTS)$ meaning that
    agents are realistic \citep{broersen2002goal};
  \item $\Conflict(\OBL,\INTS)$ meaning that agents are social \citep{jaamas:08}.
\end{itemize}

Consider the scenario of Example~\ref{ex:Jsick} with $\Aconf{\BEL}{\INT}$ and
$\Aconf{\OBL}{\INTS}$. We recall that rule $r_{5}$ states the prohibition to
visit John during the first month of his conviction. Thus, Alice has the
intention to visit John, but she does not have the social intention to do so.
This is due to rule $r_{5}$ that prevents through conflict to prove $\INTS
\mathit{visit\_John}$. At the end, it is up to the agent (or the designer of
the agent) whether to comply with the obligation, or not.}

\smallskip

\noindent The \emph{superiority relation} $>$ among rules is used to
define where one rule may override the (opposite) conclusion of another one. There are two applications of the superiority relation: the
first considers rules of the same mode while the latter compares rules of different
modes. Given $r\in R^{\Box}$ and $s\in R^{\blacksquare}$, $r>s$ iff $r$ converts
$\Box$ into $\blacksquare$ or $s$ converts $\blacksquare$ into $\Box$, i.e.,
the superiority relation is used when rules, each with a different mode, are
used to produce complementary conclusions of the same mode.
Consider the following theory
\allowdisplaybreaks
\begin{align*}
  F & = \{\mathit{go\_to\_Rome},\ \mathit{parent\_anniversary},\ \mathit{August} \} \\
  R & = \{r_1: \mathit{go\_to\_Rome} \To \mathit{go\_to\_Italy} \\ 
    & \psl  r_2: \mathit{parent\_anniversary} \To_{\OUT} \mathit{go\_to\_Rome} \\ 
    & \psl  r_3: \mathit{August} \To_{\OUT} \neg \mathit{go\_to\_Italy} \}\\
  > & = \{(r_{1}, r_{3})\}
\end{align*}
where we stipulate that $\Aconv{\Bel}{\GOAL}$ holds.

It is my parents' anniversary and they are going to celebrate it this August
in Rome, which is the capital of Italy. Typically, I do not want to go to
Italy in August since the weather is too hot and Rome itself is too crowded. 
Nonetheless, I have the goal to go to Italy this summer for my parents'
wedding anniversary, since I am a good son. Here, the superiority applies
because we use $r_{1}$ through a conversion from belief to goal.

Aligning with \citep{DBLP:journals/ai/CohenL90}, \Conflict and superiority
relations narrow and regulate the intentionality of conclusions drawn by the
\Convert relation in such a way that ``agents need not intend all the
expected side-effects of their intentions''. This also prevents the ill-famed
``dentist problem'' which brings counterintuitive consequences, as also
pointed out by \cite{DBLP:journals/ijswis/KontopoulosBGA11}. If I want to go to the dentist, either I know that the pain is a ``necessary way'' to get better, or I am a masochist. Either way, I intend to suffer some pain for getting some ends.

%


\begin{definition}[Proof]
A \emph{proof} $P$ of \emph{length} $n$ is a finite sequence $P(1), \ldots , P(n)$ of
\emph{tagged literals} of the type $+\partial_{\Box} q$ and $-\partial_{\Box} q$,
where $\Box\in\MOD$.
\end{definition}
 The proof conditions below define the logical meaning of
such tagged literals. As a conventional notation, $P(1..i)$ denotes the
initial part of the sequence $P$ of length $i$. Given a defeasible
theory $D$, $+\partial_{\Box} q$ means that $q$ is defeasibly provable in $D$
with the mode $\Box$, and $-\partial_{\Box} q$ that it has been proved in $D$
that $q$ is not defeasibly provable in $D$ with the mode $\Box$. 
%
Hereafter, the term \emph{refuted} is a synonym of \emph{not provable} and we  
use $D\vdash\pm \partial_{\Box}l$ iff there is a proof $P$ in $D$ such that 
$P(n) = \pm \partial_{\Box}l$ for an index $n$.

In order to characterise the notions of provability/refutability for
beliefs ($\pm \partial_{\BEL}$), obligations ($\pm \partial_{\OBL}$), desires
($\pm \partial_{\DES}$), goals ($\pm \partial_{\GOAL}$), intentions ($\pm
\partial_{\INT}$) and social intentions ($\pm\partial_{\INTS}$), it is
essential to define when a rule is \emph{applicable} or \emph{discarded}. To
this end, the preliminary notions of \emph{body-applicable} and
\emph{body-discarded} must be introduced. A rule is
\emph{body-applicable} when each literal in its body is proved with the
appropriate modality; a rule is \emph{body-discarded} if (at least) one of its premises has been refuted.

\begin{definition}[Body applicable]\label{def:BodyAppl}
Let $P$ be a proof and $\Box \in \set{\OBL, \DES, \GOAL, \INT, \INTS}$. A rule
$r \in R$ is \emph{body-applicable} (at $P(n+1)$) iff for all $a_i \in
A(r)$:
\begin{enumerate*}[label=(\arabic*)]
  \item if $a_i = \Box l$ then $+\partial_{\Box} l \in P(1..n)$,
  \item if $a_i = \neg \Box l$ then $-\partial_{\Box} l \in P(1..n)$, 
  \item if $a_i = l \in \LIT$ then $+\partial_{\BEL} l \in P(1..n)$.
\end{enumerate*}
\end{definition}

\begin{definition}[Body discarded]\label{def:BodyDisc}
Let $P$ be a proof and $\Box \in \set{\OBL, \DES, \GOAL, \INT, \INTS}$. A rule $r\in R$ is \emph{body-discarded} (at $P(n+1)$) iff there is
$a_{i}\in A(r)$ such that
\begin{enumerate*}[label=(\arabic*)]
  \item $a_i = \Box l$ and  $-\partial_{\Box} l\in P(1..n)$, or
  \item $a_i = \neg\Box l$ and $+\partial_{\Box}l\in P(1..n)$, or
  \item $a_i = l\in\LIT$ and $-\partial_{\BEL} l\in P(1..n)$. 
\end{enumerate*}
\end{definition}

As already stated, belief rules allow us to derive literals with different
modalities through the conversion mechanism. The applicability mechanism takes this constraint into account.

\begin{definition}[Conv-applicable]\label{def:Conv-appl}
Let $P$ be a proof. A rule $r\in R$ is \emph{Conv-applicable} (at $P(n+1)$) for $\Box$ iff
\begin{enumerate*}[label=(\arabic*)]
	\item $r\in R^{\BEL}$,
	\item $A(r)\neq\emptyset$, 
	\item $A(r)\cap\MODLIT=\emptyset$ and
	\item $\forall a\in A(r),\, +\partial_{\Box}a\in P(1..n).$
\end{enumerate*}
\end{definition}

\begin{definition}[Conv-discarded]\label{def:Conv-disc}
Let $P$ be a proof. A rule $r\in R$ is \emph{Conv-discarded} (at $P(n+1)$) for $\Box$ iff
\begin{enumerate*}[label=(\arabic*)]
	\item $r\notin R^{\BEL}$, or
	\item $A(r)=\emptyset$, or
	\item $A(r)\cap\MODLIT\neq\emptyset$, or
	\item $\exists a\in A(r)$ s.t. $-\partial_{\Box}a\in P(1..n).$
\end{enumerate*}
\end{definition}

\noindent Let us consider the following theory
\[
  F = \{ a,\ b,\ \OBL c\} \quad
  R = \{ r_1: a \To_{\OBL} b,\ r_{2}: b, c \To d\}.
\]
Rule $r_{1}$ is applicable while $r_{2}$ is not, given that $c$ is not proved
as a belief. Instead, $r_{2}$ is \emph{Conv-applicable} for $\OBL$, since
$\OBL c$ is a fact and $r_{1}$ gives $\OBL b$.

\medskip

The notion of applicability gives guidelines on how to consider the next
element in a given chain. Given that a belief rule cannot generate reparative
chains but only single literals, we conclude that the applicability condition
for belief collapses into body-applicability. When considering obligations,
each element before the current one must be a violated obligation. Concerning
desires, given that each element in an outcome chain represents a possible
desire, we only require the rule to be applicable either directly, or through
the \Convert relation. A literal is a candidate to be a goal only if none of
the previous elements in the chain has been proved as such. An intention must
pass the wishful thinking filter (that is, there is no factual knowledge for
the opposite conclusion), while social intention is also constrained not to
violate any norm.

\begin{definition}[Applicable rule]\label{def:Applicability}
	Given a proof $P$, $r\in R[q,i]$ is \emph{applicable} (at index $i$ and $P(n+1)$) for 
	\begin{enumerate}
	  \item $\BEL$ iff $r \in R^{\BEL}$ and is body-applicable.

	  \item \begin{tabbing}
		$\OBL$ iff either \= (2.1) \=(2.1.1) $r\in R^{\OBL}$ and is
	    body-applicable,\\ 
		\> \>(2.1.2) $ \forall c_k \in C(r),\, k < i,\,
	    +\partial_{\OBL}c_k \in P(1..n)$ and $-\partial c_k \in P(1..n)$, or\\
		\> (2.2) $r$ is Conv-applicable.
	  \end{tabbing}
	  \item \begin{tabbing}
		$\DES$ iff either
		\= (3.1) $r\in R^{\OUT}$ and is body-applicable, or\\
		\> (3.2) Conv-applicable.
	\end{tabbing}

	\item \begin{tabbing}
		$\Box\in\set{\GOAL,\INT,\INTS}$ iff either \= (4.1) \= (4.1.1)
	    $r\in R^{\OUT}$ and is body-applicable, and \\ 
		\> \>(4.1.2) \=$\forall c_k\in C(r), \, k<i$, $+\partial_{\blacksquare}\non c_{k}\in P(1..n)$ for some $\blacksquare$\\ 
		\> \> such that $\Conflict(\blacksquare,\Box)$ and $-\partial_{\Box}c_{k}\in P(1..n)$, or \\
		\> (4.2) $r$ is Conv-applicable. 
	\end{tabbing}
\noindent 
	\end{enumerate}
For $\GOAL$ there are no conflicts; for $\INT$ we have $\Conflict(\BEL,\INT)$, and for $\INTS$ we have $\Conflict(\BEL,\INTS)$ and $\Conflict(\OBL,\INTS)$.
\end{definition}

\begin{definition}[Discarded rule]\label{def:Discardability}
Given a proof $P$, $r\in R[q,i]$ is \emph{discarded} (at index $i$ and $P(n+1)$) for 
\begin{enumerate}
  \item $\BEL$ iff $r \in R^{\BEL}$ or is body-discarded.
  
  \item \begin{tabbing}
	$\OBL$ iff \=(2.1) \= (2.1.1) $r\notin R^{\OBL}$ or is
    body-discarded, or \\
	\> \> (2.1.2) $ \exists c_k \in C(r),\, k < i,$ s.t. $
    -\partial_{\OBL}c_k \in P(1..n)$ or $+\partial c_k \in P(1..n)$, and \\
\> (2.2) $r$ is Conv-discarded.
  \end{tabbing}
  \item \begin{tabbing}
	$\DES$ iff \= (3.1) $r\notin R^{\OUT}$ or is
    body-discarded, and\\
	\> (3.2) $r$ is Conv-discarded.
\end{tabbing}
  \item \begin{tabbing}
	$\Box\in\set{\GOAL,\INT,\INTS}$ iff \=(4.1) \=(4.1.1)
    $r\notin R^{\OUT}$ or is body-discarded, or \\
	\> \> (4.1.2) $\exists c_k\in
    C(r), \, k<i$, s.t. $-\partial_{\blacksquare}\non c_{k}\in P(1..n)$ for all $\blacksquare$ \\
	\> \> such that $\Conflict(\blacksquare,\Box)$ or $+\partial_{\Box}c_{k}\in P(1..n)$ and \\
	\> (4.2) $r$ is Conv-discarded.
\end{tabbing}
  \end{enumerate}
	For $\GOAL$ there are no conflicts; for $\INT$ we have
$\Conflict(\BEL,\INT)$, and for $\INTS$ we have $\Conflict(\BEL,\INTS)$ and
$\Conflict(\OBL,\INTS)$.
\end{definition}
Notice that the conditions of Definition~\ref{def:Discardability} are the
\emph{strong negation}\footnote{The strong negation principle is closely
related to the function that simplifies a formula by moving all negations to
an innermost position in the resulting formula, and replaces the positive
tags with the respective negative tags, and the other way around. (See
\citep{ecai2000-5,DBLP:journals/igpl/GovernatoriPRS09}.)} of those given in Definition~\ref{def:Applicability}. The
conditions to establish a rule being discarded correspond to the
constructive failure to prove that the same rule is applicable.

We are now ready to introduce the definitions of the proof conditions for the
modal operators given in this paper. We start with that for desire.

\begin{definition}[Defeasible provability for desire]\label{def:proofCond+DES}
	The proof conditions of \emph{defeasible provability} for desire are

\smallskip
\noindent
\begin{minipage}{.15\textwidth}
\begin{tabbing}
$+\partial_{\DES}$: If $P(n+1)=+\partial_{\DES} q$ then\\
  (1) \= $\DES q \in \FACTS$ or \\
  (2) \= (2.1) $\neg \DES q \not\in \FACTS$ and \\
    \> (2.2) \= $\exists r\in R[q,i]$ s.t. $r$ is applicable for $\DES$ and\\
    \> (2.3) \= $\forall s\in R[\non q,j]$ either
          (2.3.1) $s$ is discarded for $\DES$, or
          (2.3.2) $s \not > r$.
\end{tabbing}
\end{minipage}
\end{definition}

The above conditions determine when we are able to assert that $q$ is a
desire. Specifically, a \emph{desire} is each element in a
chain of an outcome rule for which there is no stronger argument for the
opposite desire.

The negative counterpart $-\partial_{\DES}q$ is obtained by the principle of strong negation.

\begin{definition}[Defeasible refutability for desire]\label{def:proofCond-DES}
	The proof conditions of \emph{defeasible refutability} for desire are

\smallskip
\noindent
\begin{minipage}{.15\textwidth}
\begin{tabbing}
$-\partial_{\DES}$: If $P(n+1)=-\partial_{\DES} q$ then\\
  (1) \= $\DES q \not\in \FACTS$ and\\
  (2) \= (2.1) $\neg \DES q \in \FACTS$, or \\
    \>   (2.2) $\forall r\in R[q,i]$ either $r$ is discarded for $\DES$, or\\
	\>   (2.3) \= $\exists s\in R[\non q,j]$ s.t.
           (2.3.1) $s$ is applicable for $\DES$ and 
           (2.3.2) $s > r$.
\end{tabbing}
\end{minipage}
\end{definition}

The proof conditions for $+\partial_{\Box}$, with $\Box \in \MOD\setminus \set{\DES}$ are as follows, provided that $\blacksquare$ and $\Diamond$ represent two arbitrary modalities in $\MOD$:

\begin{definition}[Defeasible provability for obligation, goal, intention and social intention]\label{def:proofCond+X}
	The proof conditions of \emph{defeasible provability} for $\Box \in \MOD\setminus \set{\DES}$ are
	
	\smallskip
	\noindent
	\begin{minipage}{.15\textwidth}
	\begin{tabbing}
	$+\partial_{\Box}$: If $P(n+1)=+\partial_{\Box} q$ then\\
	  (1) \= $\Box q \in \FACTS$ or \\
	  (2) \= (2.1) $\neg \Box q \not\in \FACTS$ and $(\blacksquare \non q \not\in \FACTS$ for $\blacksquare = \Box$ or $\Conflict(\blacksquare, \Box))$ and \\
	    \> (2.2) \= $\exists r\in R[q,i]$ s.t. $r$ is applicable for $\Box$ and\\
	    \> (2.3) \= $\forall s\in R[\non q,j]$ either\\
	    \> \> (2.3.1) $\forall \blacksquare$ s.t. $\blacksquare=\Box$ or $\Aconf{\blacksquare}{\Box}$, $s$ is discarded for $\blacksquare$; or \\
	    \> \> (2.3.2) \=$\exists \Diamond, \exists t\in R[q,k]$ s.t. $t$ is applicable for $\Diamond$, and either\\
	                \> \> \> (2.3.2.1) \= $t>s$ if $\blacksquare = \Diamond$, $\Aconv{\blacksquare}{\Diamond}$, or $\Aconv{\Diamond}{\blacksquare}$; or\\
	    \> \> \> (2.3.2.2) \= $\Aconf{\Diamond}{\blacksquare}$.
	\end{tabbing}
	\end{minipage}
\end{definition}
To show that a literal $q$ is defeasibly provable with the modality $\Box$ we
have two choices: (1) the modal literal $\Box q$ is a fact; or (2) we need to
argue using the defeasible part of $D$. For (2), we require that (2.1) a
complementary literal (of the same modality, or of a conflictual modality)
does not appear in the set of facts, and (2.2) there must be an applicable
rule for $\Box$ and $q$. Moreover, each possible attack brought by a rule $s$
for $\non q$ has to be either discarded for the same modality of $r$ and for
all modalities in conflict with $\Box$ (2.3.1), or successfully
counterattacked by another stronger rule $t$ for $q$ (2.3.2). We recall that
the superiority relation combines rules of the same mode, rules with
different modes that produce complementary conclusion of the same mode
through conversion (both considered in clause (2.3.2.1)), and rules with
conflictual modalities (clause 2.3.2.2). Trivially, if $\Box=\BEL$ then the
proof conditions reduce to those of classical defeasible logic \citep{tocl}.

Again, conditions for $-\partial_{\Box}$ are derived by the principle of
strong negation from that for $+\partial_{\Box}$ and are as follows.

\begin{definition}[Defeasible refutability for obligation, goal, intention and social intention]\label{def:proofCond-X}
	The proof conditions of \emph{defeasible refutability} for $\Box \in\set{\OBL, \GOAL, \INT, \INTS}$ are

\smallskip	
\noindent
\begin{minipage}{.15\textwidth}
\begin{tabbing}
$-\partial_{\Box}$: If $P(n+1)=-\partial_{\Box} q$ then\\
  (1) \= $\Box q \notin \FACTS$ and either\\
  (2) \= (2.1) $\neg \Box q \in \FACTS$ or $(\blacksquare \non q \in \FACTS$ for $\blacksquare=\Box$ or $\Conflict(\blacksquare,\Box))$ or \\
    \> (2.2) \= $\forall r\in R[q,i]$ either $r$ is discarded for $\Box$ or\\
    \> (2.3) \= $\exists s\in R[\non q,j]$ s.t.\\
    \> \> (2.3.1) $\exists \blacksquare$ s.t. $(\blacksquare = \Box$ or $\Aconf{\blacksquare}{\Box})$ and $s$ is applicable for $\blacksquare$, and \\
    \> \> (2.3.2) \= $\forall \Diamond, \forall t\in R[q,k]$ either $t$ is discarded for $\Diamond$, or\\
    	\> \> \> (2.3.2.1) \= $t\not>s$ if $\blacksquare= \Diamond$, $\Aconv{\blacksquare}{\Diamond}$, or $\Aconv{\Diamond}{\blacksquare}$; and\\ 
        \> \> \> (2.3.2.2) not $\Conflict(\Diamond,\blacksquare)$.
\end{tabbing}
\end{minipage}
\end{definition}

To better understand how applicability and proof conditions interact to define the (defeasible) conclusions of a given theory, we consider the example below.

\begin{example}\label{ex:proofTagsExplanation}
Let $D$ be the following modal theory	
  \begin{align*}
    F & = \set{ a_{1},\, a_{2},\, \neg b_{1},\, \OBL \neg b_{2}} &
    R & = \set{ r: {a_{1}}\To_{\OUT} b_{1}\odot b_{2}\odot b_{3}\odot b_{4},
	\ s: a_{2} \To_{\OUT} b_{4}}.
  \end{align*}
Here, $r$ is trivially applicable for $\DES$ and $+\partial_{\DES} b_{i}$
holds, for $1\leq i \leq 4$. Moreover, we have $+\partial_{\GOAL} b_{1}$ and
$r$ is discarded for $\GOAL$ after $b_{1}$. Due to $+\partial \neg b_{1}$, it
follows that $-\partial_{\INT}b_{1}$ holds (as well as
$-\partial_{\INTS}b_{1}$); the rule is applicable for $\INT$ and $b_{2}$, and
we are able to prove $+\partial_{\INT}b_{2}$; the rule is thus discarded for
$\INT$ and $b_{3}$ as well as $b_{4}$. Due to $\OBL\neg b_{2}$ being a fact,
$r$ is discarded for $\INTS$ and $b_{2}$ resulting in
$-\partial_{\INTS}b_{2}$, which in turn makes the rule applicable for $\INTS$
and $b_{3}$, proving $+\partial_{\INTS}b_{3}$. As we have argued before, this
makes $r$ discarded for $b_{4}$. Even if $r$ is discarded for $\INTS$ and
$b_{4}$, we nonetheless have $D\vdash +\partial_{\INTS}b_{4}$ due to $s$;
specifically, $D\vdash +\partial_{X}b_{4}$ with $X\in \set{\DES, \GOAL, \INT,
\INTS}$ given that $s$ is trivially applicable for $X$.

For further illustrations of how the machinery works, the reader is referred to \ref{sec:table}.
\end{example}

The next definition extends the concept of complement for modal literals and
is used to establish the logical connection among proved and refuted literals
in our framework.

\begin{definition}[Complement set]\label{def:complement}
The \emph{complement set} of a given modal literal $l$, denoted by
$\tilde{l}$, is defined as follows:
\begin{enumerate*}[label=(\arabic*)]
\item if $l=\DES m$, then $\tilde{l}=\set{\neg\DES m}$;
\item if $l=\Box m$, then $\tilde{l}=\set{\neg\Box m, \Box\non m}$, with $\Box \in \set{\OBL, \GOAL, \INT, \INTS}$;
\item if $l=\neg \Box m$, then $\tilde{l}=\set{\Box m}$.
\end{enumerate*}
\end{definition}

The logic resulting from the above proof conditions enjoys
properties describing the appropriate behaviour of the modal operators for consistent theories.
\begin{definition}[Consistent defeasible theory]\label{def:consistency}
A defeasible theory $D = (\FACTS, R, >)$ is \emph{consistent} iff $>$ is acyclic and $\FACTS$ does not contain pairs of complementary literals, that is if $\FACTS$ does not contain pairs like (i) $l$ and $\non l$, (ii) $\Box l$ and $\neg\Box l$ with $\Box\in\MOD$, and (iii) $\Box l$ and $\Box\non l$ with $\Box\in \set{\GOAL, \INT, \INTS}$.
\end{definition}

\begin{restatable}{prop}{Coherence}
\label{prop:CoherenceConsistence}
Let $D$ be a consistent, finite defeasible theory. For any literal $l$, it is not possible to have both 
\begin{enumerate}
	\item $D\vdash+\partial_{\Box}l$ and $D\vdash-\partial_{\Box}l$ with $\Box\in\MOD$;
	 \item $D\vdash+\partial_{\Box}l$ and $D\vdash+\partial_{\Box}\non l$ with $\Box\in \MOD\setminus\set{\DES}$.
\end{enumerate}
\end{restatable}
All proofs of propositions, lemmas and theorems are reported in \ref{sec:proofs} and \ref{sec:CorrDefExt}.

The meaning of the above proposition is that, for instance, it is not
possible for an agent to obey something that is obligatory and forbidden (obligatory not)
at the same time. On the other hand, an agent may have opposite desires given
different situations, but then she will be able to plan for only one between
the two alternatives.

Proposition \ref{prop:+PartialTHEN-Partial} below governs the interactions between different modalities
and the relationships between proved literals and refuted complementary
literals of the same modality. Proposition \ref{prop:donothold} proves that certain (likely-expected) implications do no hold.

\begin{restatable}{prop}{PositiveProp}\label{prop:+PartialTHEN-Partial}
Let $D$ be a consistent defeasible theory. For any literal $l$, the following statements hold:
\begin{enumerate}
	\item if $D\vdash+\partial_{\Box} l$, then $D\vdash-\partial_{\Box} \non l$ with $\Box \in \MOD \setminus \set{\DES}$;
	\item if $D\vdash+\partial l$, then $D\vdash -\partial_{\INT} \non l$; 
	\item if $D\vdash+\partial l$ or $D\vdash+\partial_{\OBL} l$, then $D\vdash -\partial_{\INTS} \non l$;
	\item\label{enum5} if $D\vdash +\partial_{\GOAL} l$, then $D\vdash +\partial_{\DES} l$;
	\item if $D\vdash -\partial_{\DES}l$, then $D\vdash -\partial_{\GOAL}l$.
\end{enumerate}
\end{restatable}

\begin{restatable}{prop}{NegativeProp}\label{prop:donothold}
Let $D$ be a consistent defeasible theory. For any literal $l$, the following statements \emph{do not} hold:
\begin{enumerate} \setcounter{enumi}{5}
	\item if $D\vdash +\partial_{\DES} l$, then $D\vdash +\partial_{X} l$ with $X\in \set{\GOAL, \INT, \INTS}$;	
	\item if $D\vdash +\partial_{\GOAL} l$, then $D\vdash +\partial_{X} l$ with $X\in \set{\INT, \INTS}$;
	\item if $D\vdash +\partial_{X} l$, then $D\vdash +\partial_{Y} l$ with $X=\set{\INT, \INTS}$ and $Y=\set{\DES, \GOAL}$;	
	\item if $D\vdash -\partial_{Y} l$, then $D\vdash -\partial_{X} l$ with $Y\in \set{\DES, \GOAL}$ and $X\in \set{\INT, \INTS}$.
\end{enumerate}
\end{restatable}

Parts 6. and 7. directly follow by Definitions from~\ref{def:Applicability} to \ref{def:proofCond-X} and rely on the intuitions
presented in Section~\ref{sec:Intuition}. Parts from 7. to 9.
reveal the true nature of expressing outcomes in a preference order: it may be
the case that the agent desires something (may it be even her preferred
outcome) but if the factuality of the environment makes this outcome impossible
to reach, then she should not pursue such an outcome, and instead commit
herself on the next option available.
The statements of Proposition \ref{prop:donothold} exhibit a common
feature which can be illustrated by the idiom: ``What's your plan B?'',
meaning: even if you are willing for an option, if such an option is not
feasible you need to strive for the plan B. 

%
%
%

\let\Box\boxbox
\let\blacksquare\boxblacksquare
\let\Diamond\boxdiamond

\section{Algorithmic results} 
\label{sec:algorithmic_results}

We now present procedures and algorithms to compute the
\emph{extension} of a \emph{finite} defeasible theory
(Subsection~\ref{subsec:algorithms}), in order to ascertain the complexity of
the logic introduced in the previous sections. The algorithms are inspired to
ideas proposed in \citep{Maher2001,Lam.2011}.
%

\subsection{Notation for the algorithms}\label{subsec:notation}

From now on, $\blacksquare$ denotes  a generic modality
in \MOD, $\Diamond$ a generic modality in $\MOD \setminus \set{\BEL}$, and
$\Box$ a fixed modality chosen in $\blacksquare$. Moreover, whenever $\Box = \BEL$ we shall treat
literals $\Box l$ and $l$ as synonyms. To accommodate
the \Convert relation to the algorithms, we recall that $R^{\BEL,\Diamond}$
denotes the set of belief rules that can be used for a conversion to modality
$\Diamond$. The antecedent of all such rules is not empty, and does
not contain any modal literal.

Furthermore, for each literal $l$, $l_{\blacksquare}$ is the set (initially
empty) such that $\pm\Box \in l_{\blacksquare}$ iff $D\vdash\pm\partial_{\Box}
l$. Given a modal defeasible theory $D$, a set of rules $R$, and a rule $r\in
R^{\Box}[l]$, we expand the superiority relation $>$ by incorporating the
\Conflict relation into it:
\[
> = > \cup\,\set{(r,s) | r \in R^{\Box}[l], s \in R^{\blacksquare}[\non l], \Aconf{\Box}{\blacksquare}}.
\]
We also define:
\begin{enumerate}
\item $r_{sup}=\set{s \in R :
(s,r)\in >}$ and $r_{inf}=\set{s \in R : (r,s)\in >}$ for any $r\in R$; 
\item
$HB_{D}$ as the set of literals such that the literal or its complement
appears in $D$, i.e., such that it is a sub-formula of a modal
literal occurring in $D$; 
\item  the modal Herbrand Base of $D$ as
$HB=\set{\Box l |\ \Box \in\MOD, l\in HB_{D}}$. 
\end{enumerate}
Accordingly, the extension of a defeasible theory is defined as follows.
\begin{definition}[Defeasible extension]\label{def:extension}
Given a defeasible theory $D$, the \emph{defeasible extension} of $D$ is
defined as
\[
  E(D) = (+\partial_{\Box}, -\partial_{\Box}),
\]
where $\pm\partial_{\Box} = \set{l \in HB_{D}: D \vdash
\pm\partial_{\Box} l}$ with $\Box \in \MOD$. 
Two defeasible theories $D$ and $D'$ are \emph{equivalent} whenever they have the same extensions, i.e., $E(D) = E(D')$.
\end{definition}

We introduce two operations that modify the consequent of
rules used by the algorithms.

\begin{definition}[Truncation and removal]\label{def:trunctaion-removal}
  Let $c_1=\opseq[\odot]{a}{1}{i-1}$ and $c_2=\opseq[\odot]{a}{i+1}{n}$
  be two (possibly empty) $\odot$-expressions such that $a_{i}$ does not occur
  in neither of them, and $c=c_{1}\odot a_{i}\odot c_{2}$ is an $\odot$-expression.
  Let $r$ be a rule with form $A(r)\To_{\Diamond} c$. We define
  the \emph{truncation} of the consequent $c$ at $a_{i}$ as:
  \[
    A(r)\To_{\Diamond} c!a_{i} = A(r)\To_{\Diamond} c_{1}\odot a_{i},
  \]
and the \emph{removal} of $a_{i}$ from the consequent $c$ as:
\begin{displaymath}
	A(r)\To_{\Diamond} c\ominus a_{i} = A(r) \To_{\Diamond} c_{1} \odot c_{2}.
\end{displaymath}
\end{definition}

Notice that removal may lead to rules with empty consequent which
strictly would not be rules according to the definition of the language.
Nevertheless, we accept such expressions within the description of the algorithms but then such rules will not be in any $R[q,i]$ for any $q$ and
$i$. In such cases, the operation \emph{de facto} removes the rules.

Given $\Box\in \MOD$, the sets $+\partial_{\Box}$ and $-\partial_{\Box}$
denote, respectively, the global sets of positive and negative defeasible
conclusions (i.e., the set of literals for which condition $+\partial_{\Box}$
or $-\partial_{\Box}$ holds), while $\partial_{\Box}^{+}$ and
$\partial_{\Box}^-$ are the corresponding temporary sets, that is the set
computed at each iteration of the main algorithm. Moreover, to simplify the
computation, we do not operate on outcome rules: for each rule $r \in R^{\OUT}$
we create instead a new rule for desire, goal, intention, and social
intention (respectively, $r^{\DES}$, $r^{\GOAL}$, $r^{\INT}$, and
$r^{\INTS}$). Accordingly, for the sake of simplicity, in the present section
we shall use expressions like ``the intention rule'' as a shortcut for ``the
clone of the outcome rule used to derive intentions''.

\subsection{Algorithms} 
\label{subsec:algorithms}

The idea of all the algorithms is to use the operations of truncation
and elimination to obtain, step after step, a simpler
but equivalent theory. In fact, proving a literal does not give local
information regarding the element itself only, but rather reveals which rules
should be discarded, or reduced, in their head or body. Let us assume that, at a given step, the algorithm proves literal $l$. At the next step,
\begin{enumerate}
	\item the applicability of any rule $r$ with $l\in A(r)$
does not depend on $l$ any longer. Hence, we can safely remove $l$ from $A(r)$.
	\item Any rule $s$ with $\widetilde{l} \cap A(s)\neq\emptyset$ is
discarded. Consequently, any superiority tuple involving $s$ is now
useless and can be removed from the superiority relation.
	\item We can shorten chains by exploiting conditions of
Definitions~\ref{def:Applicability} and \ref{def:Discardability}. For instance, if $l=\OBL m$, we can
truncate chains for obligation rules at $\non m$ and eliminate it as
well.
\end{enumerate}

\begin{algorithm}[htb]
\fontsize{8}{9.5}\selectfont
\caption{\textsc{DefeasibleExtension}}\label{alg:defeasible}
\begin{algorithmic}[1]

	
\State $+\partial_{\blacksquare}, \partial^{+}_{\blacksquare}\gets \emptyset;\  -\partial_{\blacksquare}, \partial^{-}_{\blacksquare}\gets \emptyset$\label{CDriga1}


\State $R \gets R \cup \set{r^{\Box}: A(r)\To_{\Box} C(r) | r \in R^{\OUT}}$, with $\Box \in \set{\DES, \GOAL, \INT, \INTS}$\label{CDrigaOUTCopy}
\State $R \gets R \setminus R^{\OUT}$\label{CDrigaR-}

\State $R^{\BEL,\Diamond}\gets \set{r^{\Diamond}: A(r) \hookrightarrow C(r) | r\in R^{\BEL}, 
     A(r)\neq \emptyset, A(r)\subseteq \LIT}$\label{CDrigaRconv}
\State $> \gets > \cup \set{(r^{\Diamond},s^{\Diamond})| r^{\Diamond}, s^{\Diamond}\in R^{\BEL,\Diamond}, r>s} \cup \set{(r, s) | r \in R^{\blacksquare} \cup R^{\BEL, \blacksquare}, s\in R^{\Diamond} \cup R^{\BEL, \Diamond}, \Conflict(\blacksquare,\Diamond) }$\label{CDrigaSup1}

\For{$l \in F$}\label{CDrigaFor2}
	\State \Ifline{$l = \Box m$}{\Call{Proved}{$m,\,\Box$}}\label{CDrigaIfMM}
	\State \Ifline{$l = \neg \Box m \wedge \Box\neq\DES$}{\Call{Refuted}{$m,\,\Box$}}\label{CDrigaIfDisc}	
\EndFor\label{CDrigaEndFor2}

\State $+\partial_{\blacksquare}\gets +\partial_{\blacksquare}\cup\partial^{+}_{\blacksquare};\  -\partial_{\blacksquare}\gets -\partial_{\blacksquare}\cup\partial^{-}_{\blacksquare}$\label{CDrigaUp}
\State $R_{\infd} \gets \emptyset$\label{CDrigaPulisci}

\Repeat\label{CDrigaRepeat}
  \State $\partial^{+}_{\blacksquare} \gets \emptyset;\ \partial^{-}_{\blacksquare} \gets \emptyset$\label{CDrigainit}
  \For{$\Box l\in HB$}\label{CDrigaRepeatFor1}
    \State \Ifline{$R^{\Box}[l]\cup R^{\BEL,\Box}[l]=\emptyset$}{\Call{Refuted}{$l,\,\Box$}}\label{CDrigaIfDisc2}
  \EndFor\label{CDrigaEndRepeatFor1}
  \For{$r\in R^{\Box}\cup R^{\BEL,\Box}}$\label{CDrigaRepeatFor2}
	\If{$A(r)=\emptyset$}\label{CDrigaMainRepeatIf}
	 	\State $r_{inf} \gets \set{r \in R: (r, s) \in >, s\in R}$; $r_{sup} \gets \set{s \in R : (s,r)\in >}$\label{CDrigarInfSup}
    	\State $R_{\infd} \gets R_{\infd} \cup r_{inf}$\label{CDrigaRinf}
    	\State Let $l$ be the first literal of $C(r)$ in $HB$\label{CDrigaLet}
		\If{$r_{sup} = \emptyset$}\label{CDrigaIfIf}
			\If{$\Box = \DES$}\label{CDrigaIfDES}
				\State \Call{Proved}{$m,\,\DES$}\label{CDrigaProvedDES}
			\Else	
			\State \Call{Refuted}{$\non l,\,\Box$}\label{CDrigaRepeatDisc1}
			\State \Call{Refuted}{$\non l,\,\Diamond$} for $\Diamond$ s.t.
          $\Aconf{\Box}{\Diamond}$\label{CDrigaRepeatDisc2}
    		\If{$R^{\Box}[\non l] \cup R^{\BEL,\Box}[\non l] \cup R^{\blacksquare}[\non l]\setminus
    		          R_{\infd}\subseteq r_{inf}$,  for $\blacksquare$ s.t.
    		          $\Aconf{\blacksquare}{\Box}$}\label{CDrigaIfIfIf}
				\State \Call{Proved}{$m,\,\Box$}\label{CDrigaModRepeat}
      		\EndIf\label{CDrigaEndIfIfIf}
			\EndIf\label{CDrigaEndIfDES}
    	\EndIf\label{CDrigaEndIfIf}
  	\EndIf\label{CDrigaEndMainRepeatIf}
  \EndFor\label{CDrigaEndRepeatFor2}

\State $\partial^{+}_{\blacksquare}\gets \partial^{+}_{\blacksquare}\setminus +\partial_{\blacksquare};\ \partial^{-}_{\blacksquare}\gets \partial^{-}_{\blacksquare}\setminus -\partial_{\blacksquare}$\label{CDrigaFinalUp1}
  \State $+\partial_{\blacksquare}\gets +\partial_{\blacksquare}\cup\partial^{+}_{\blacksquare};\ -\partial_{\blacksquare}\gets -\partial_{\blacksquare}\cup\partial^{-}_{\blacksquare}$\label{CDrigaFinalUp2}
\Until{$\partial^{+}_{\blacksquare}=\emptyset$ and $\partial^{-}_{\blacksquare}=\emptyset$}\label{CDrigaUntil}

\State \Return $(+\partial_{\blacksquare}, -\partial_{\blacksquare})$\label{CDrigaReturn}
\end{algorithmic}
\end{algorithm}

Algorithm~\ref{alg:defeasible}~\textsc{DefeasibleExtension} is the core
algorithm to compute the extension of a defeasible theory. The first part of
the algorithm (lines \ref{CDriga1}--\ref{CDrigaSup1}) sets up the data
structure needed for the computation. Lines
\ref{CDrigaFor2}--\ref{CDrigaEndFor2} are to handle facts as immediately
provable literals.

The main idea of the algorithm is to check whether there are rules with empty
body: such rules are clearly applicable and they can produce conclusions with
the right mode. However, before asserting that the first element for the
appropriate modality of the conclusion is provable, we need to check whether
there are rules for the complement with the appropriate mode; if so, such
rules must be weaker than the applicable rules. The information about which
rules are weaker than the applicable ones is stored in the support set
$R_{\mathit{infd}}$. When a literal is evaluated to be provable, the
algorithm calls procedure \textsc{Proved}; when a literal is rejected,
procedure \textsc{Refuted} is invoked. These two procedures apply
transformations to reduce the complexity of the theory.

A step-by-step description of the algorithm would be redundant once the
concepts expressed before are understood. Accordingly, in the rest 
of the section we provide in depth descriptions of the key passage. 

For every outcome rule, the algorithm makes a copy of the same rule for each
mode corresponding to a goal-like attitude (line \ref{CDrigaOUTCopy}). At
line \ref{CDrigaRconv}, the algorithm creates a support set to handle
conversions from a belief rule through a different mode. Consequently, the
new $\Diamond$ rules have to inherit the superiority relation (if any) from
the belief rules they derive from (line \ref{CDrigaSup1}). Notice that we
also augment the superiority relation by incorporating the rules involved in
the \Conflict relation. Given that facts are immediately proved literals,
\textsc{Proved} is invoked for positively proved modal literals (those
proved with $+\partial_{\Box}$), and \textsc{Refuted} for rejected literals
(i.e., those proved with $-\partial_{\Box}$). The aim of the \textbf{for}
loop at lines \ref{CDrigaRepeatFor1}--\ref{CDrigaEndRepeatFor1} is to discard
any modal literal in $HB$ for which there are no rules that can prove it
(either directly or through conversion).

We now iterate on every rule that can fire (i.e., on rules with empty body, loop
\textbf{for} at lines \ref{CDrigaRepeatFor2}--\ref{CDrigaEndRepeatFor2} and
\textbf{if} condition at line \ref{CDrigaMainRepeatIf}) and we collect the
weaker rules in the set $R_{\infd}$ (line \ref{CDrigaRinf}). Since a
consequent can be an $\odot$-expression, the literal we are interested in is
the first element of the $\odot$-expression (line \ref{CDrigaLet}). If no
rule stronger than the current one exists, then the complementary conclusion
is refuted by condition (2.3) of Definition~\ref{def:proofCond-X} (line
\ref{CDrigaRepeatDisc1}). An additional consequence is that literal $l$ is
also refutable in $D$ for any modality conflicting with $\Box$ (line
\ref{CDrigaRepeatDisc2}). Notice that this reasoning does not hold for
desires: since the logic allows to have $\DES l$ and $\DES \non l$ at the
same time, when $\Box = \DES$ and the guard at line \ref{CDrigaIfIf} is
satisfied, the algorithm invokes procedure~\ref{alg:proved}
\textsc{Proved} (line \ref{CDrigaProvedDES}) due to condition (2.3) of
Definition \ref{def:proofCond+DES}.

The next step is to check whether there are rules for the complement
literal of the same modality, or of a conflicting modality. The rules for the
complement should not be defeated by applicable rules: such rules thus cannot be in $R_{\infd}$. If all these rules are defeated by $r$ (line
\ref{CDrigaIfIfIf}), then conditions for deriving $+\partial_\Box$ are
satisfied, and Algorithm~\ref{alg:proved} \textsc{Proved} is invoked.


\begin{algorithm}[htb]
  \fontsize{8}{9.5}\selectfont
\caption{\textsc{Proved}}\label{alg:proved}
\begin{algorithmic}[1]
\Procedure{Proved}{$l \in \LIT,\,\Box \in \MOD$}\label{MMriga}

\State $\partial^+_{\Box}\gets \partial^+_{\Box} \cup \set{l};\ l_{\blacksquare} \gets l_{\blacksquare} \cup \set{+\Box}$\label{MMrigaUpLblack}
\State $HB \gets HB \setminus \set{\Box l}$\label{MMrigaUpHB}

\State \Ifline{$\Box \neq \DES$}{\Call{Refuted}{$\non l,\,\Box$}}\label{MMrigaIfD}
\State \Ifline{$\Box = \BEL$}{\Call{Refuted}{$\non l,\,\INT$}}\label{MMrigaIfBELpINTnon-p}
\State \Ifline{$\Box \in\set{ \BEL, \OBL}$}{\Call{Refuted}{$\non l,\,\INTS$}}\label{MMrigaIfBEL-OBLpINTSnon-p}

\State $R \gets \set{r: A(r) \setminus \set{\Box l, \neg \Box \non l}\hookrightarrow C(r) |\ r\in R,\ A(r) \cap \widetilde{\Box l}  = \emptyset}$\label{MMrigaR}

\State $R^{\BEL,\Box} \gets \set{r: A(r)\setminus \set{l}\hookrightarrow C(r)| r\in R^{\BEL,\Box},\ \non l \notin A(r)}$\label{MMrigaRConv}

\State $> \gets > \setminus \set{(r,s), (s,r) \in > |\  A(r) \cap \widetilde{\Box l} \not = \emptyset}$\label{MMrigaSup}

\Switch{$\Box$}\label{MMrigaSwitch}

	\Case{$\BEL$}\label{MMrigaCaseB}
		\State $R^{X}\gets \set{ A(r) \To_{X} C(r)! l |\ r\in R^{X}[l,n]}$ with $X\in \set{\OBL,\INT}$\label{MMrigaBR1}
		\State \Ifline{$+\OBL\in \non l_{\blacksquare}$}{$R^{\OBL}\gets \set{ A(r) \To_{\OBL} C(r) \ominus \non l |\ r\in R^{\OBL}[\non l,n]}$}\label{MMrigaBIf+OR}
		\State \Ifline{$-\OBL\in \non l_{\blacksquare}$}{$R^{\INTS}\gets \set{ A(r) \To_{\INTS} C(r) ! l |\ r\in R^{\INTS}[l,n]}$}\label{MMrigaBIf-OR}

	\Case{$\OBL$}\label{MMrigaCaseO}
		\State $R^{\OBL}\gets \set{ A(r) \To_{\OBL} C(r)! \non l \ominus \non l |\ r\in R^{\OBL}[\non l,n]}$\label{MMrigaOR1}
		\State \Ifline{$-\BEL\in l_{\blacksquare}$}{$R^{\OBL}\gets \set{ A(r) \To_{\OBL} C(r) \ominus l |\ r\in R^{\OBL}[l,n]}$}\label{MMrigaOIfR1}
		\State \Ifline{$-\BEL\in \non l_{\blacksquare}$}{$R^{\INTS}\gets \set{ A(r) \To_{\INTS} C(r) ! l |\ r\in R^{\INTS}[l,n]}$}\label{MMrigaOIfR2}
		
	\Case{$\DES$}\label{MMrigaCaseD}
		\If{$+\DES\in \non l_{\blacksquare}$}\label{MMrigaDIf}
			\State $R^{\GOAL}\gets \set{ A(r) \To_{\GOAL} C(r)! l \ominus l |\ r\in R^{\GOAL}[l,n]}$\label{MMrigaDIfR1}
			\State $R^{\GOAL}\gets \set{ A(r) \To_{\GOAL} C(r)! \non l \ominus\non l  |\ r\in R^{\GOAL}[\non l,n]}$\label{MMrigaDIfR2}
		\EndIf\label{MMrigaEndDIf}
		
	\Otherwise\label{MMrigaOtherwise}
		\State $R^{\Box}\gets \set{ A(r) \To_{\Box} C(r)! l |\ r\in R^{\Box}[l,n]}$\label{MMrigaOtherR1}
		\State $R^{\Box}\gets \set{ A(r) \To_{\Box} C(r) \ominus \non l |\ r\in R^{\Box}[\non l,n]}$\label{MMrigaOtherR2}		
	\EndOther\label{MMrigaEndOtherwise}
	
\EndProcedure\label{MMrigaEnd}
\end{algorithmic}
\end{algorithm}

\smallskip

\noindent Algorithm~\ref{alg:proved}~\textsc{Proved} is invoked when literal $l$ is
proved with modality $\Box$, the key to which
simplifications on rules can be done. The computation starts by updating the
relative positive extension set for modality $\Box$ and, symmetrically, the
local information on literal $l$ (line \ref{MMrigaUpLblack}); $l$ is then
removed from $HB$ at line \ref{MMrigaUpHB}. Parts 1.--3. of
Proposition~\ref{prop:+PartialTHEN-Partial} identifies the modalities literal
$\non l$ is refuted with, when $\Box l$ is proved (\textbf{if} conditions at
lines \ref{MMrigaIfD}--\ref{MMrigaIfBEL-OBLpINTSnon-p}). Lines \ref{MMrigaR}
to \ref{MMrigaSup} modify the superiority relation and the sets of rules $R$
and $R^{\BEL,\Box}$ accordingly to the intuitions given at the beginning of Section \ref{subsec:algorithms}.

Depending on the modality $\Box$ of $l$, we perform specific
operations on the chains (condition \textbf{switch} at lines
\ref{MMrigaSwitch}--\ref{MMrigaEndOtherwise}). A detailed description of each
\textbf{case} would be redundant without giving more information than the one
expressed by conditions of Definitions~\ref{def:Applicability} and
\ref{def:Discardability}. Therefore, we propose one significative example by
considering the scenario where $l$ has been proved as a belief (\textbf{case}
at lines \ref{MMrigaCaseB}--\ref{MMrigaBIf-OR}). First, conditions of
Definitions~\ref{def:Discardability} and \ref{def:proofCond-X} ensure that
$\non l$ may be neither an intention, nor a social intention.
Algorithm~\ref{alg:refuted} \textsc{Refuted} is thus invoked at lines
\ref{MMrigaIfBELpINTnon-p} and \ref{MMrigaIfBEL-OBLpINTSnon-p} which, in
turn, eliminates $\non l$ from every chain of intention and social intention
rules (line \ref{DrigaCaseOtherR} of Algorithm~\ref{alg:refuted}
\textsc{Refuted}). Second, chains of obligation (resp. intention) rules can
be truncated at $l$ since condition (2.1.2) (resp. condition (4.1.2)) of
Definition~\ref{def:Discardability} makes such rules discarded for all
elements following $l$ in the chain (line \ref{MMrigaBR1}). Third, if
$+\partial_{\OBL} \non l$ has been already proved, then we eliminate $\non l$
in chains of obligation rules since it represents a violated obligation
(\textbf{if} condition at lines \ref{MMrigaBIf+OR}). Fourth, if
$-\partial_{\OBL} \non l$ is the case, then each element after $l$ cannot be
proved as a social intention (\textbf{if} condition at line
\ref{MMrigaBIf-OR}). Consequently, we truncate chains of social intention
rules at $l$.

\begin{algorithm}
\fontsize{8}{9.5}\selectfont
\caption{\textsc{Refuted}}\label{alg:refuted}
\begin{algorithmic}[1]
\Procedure{Refuted}{$l\in \LIT,\,\Box \in \MOD$}\label{Driga}

\State $\partial^-_\Box\gets \partial^-_\Box\cup\set{l};\ l_{\blacksquare} \gets l_{\blacksquare} \cup \set{-\Box}$\label{DrigaUpL}
\State $HB \gets HB \setminus \set{\Box l}$\label{DrigaHB}
\State $R \gets \set{r: A(r)\setminus\set{\neg\Box l}\hookrightarrow C(r)|\ r\in R,\ \Box l \not\in A(r)}$\label{DrigaCaseUpR2}
\State $R^{\BEL,\Box} \gets R^{\BEL,\Box} \setminus \set{r\in R^{\BEL,\Box}: l\in A(r)}$\label{DrigaCaseUpConv}
\State $> \gets > \setminus \set{(r,s),(s,r)\in>| \Box l\in A(r)}$\label{DrigaCaseUpSup}

\Switch{$\Box$}\label{DrigaSwitch}

	\Case{$\BEL$}\label{DrigaCaseB}
		\State $R^{\INT} \gets \set{A(r) \To_{\INT} C(r)! \non l | r\in R^{\INT}[\non l,n]}$\label{DrigaCaseRI}
		\State \Ifline{$+\OBL\in l_{\blacksquare}$}{$R^{\OBL} \gets \set{A(r) \To_{\OBL} C(r) \ominus l| r\in R^{\OBL}[l,n]}$}\label{DrigaCaseBIf+OR}
		\State \Ifline{$-\OBL\in l_{\blacksquare}$}{$R^{\INTS} \gets \set{A(r) \To_{\INTS} C(r)! \non l | r\in R^{\INTS}[\non l,n]}$}\label{DrigaCaseBIf-OR}

	\Case{$\OBL$}\label{DrigaCaseO}
		\State $R^{\OBL} \gets \set{A(r) \To_{\OBL} C(r)! l \ominus l| r\in R^{\OBL}[l,n]}$\label{DrigaCaseOR}
		\State \Ifline{$-\BEL\in l_{\blacksquare}$}{$R^{\INTS} \gets \set{A(r) \To_{\INTS} C(r)! \non l | r\in R^{\INTS}[\non l,n]}$}\label{DrigaCaseOIfR}

	\Case{$\DES$}\label{DrigaCaseD}
		\State $R^{X} \gets \set{A(r) \To_{X} C(r) \ominus l | r\in R^{X}[l,n]}$ with $X\in \set{\DES,\GOAL}$\label{DrigaCaseDR}

	\Otherwise\label{DrigaOtherwise}
		\State $R^{\Box} \gets \set{A(r) \To_{\Box} C(r) \ominus l | r\in R^{\Box}[l,n]}$\label{DrigaCaseOtherR}
\EndOther\label{DrigaEndOtherwise}

\EndProcedure\label{DrigaEnd}
\end{algorithmic}
\end{algorithm}

\noindent Algorithm~\ref{alg:refuted}~\textsc{Refuted} performs all necessary
operations to refute literal $l$ with modality $\Box$. The initialisation
steps at lines \ref{DrigaUpL}--\ref{DrigaCaseUpSup} follow the same schema
exploited at lines \ref{MMrigaUpLblack}--\ref{MMrigaSup} of
Algorithm~\ref{alg:proved}~\textsc{Proved}. Again, the operations on chains
vary according to the current mode $\Box$ (\textbf{switch} at lines
\ref{DrigaSwitch}--\ref{DrigaEndOtherwise}). For instance, if $\Box = \BEL$
(\textbf{case} at lines \ref{DrigaCaseB}--\ref{DrigaCaseBIf-OR}), then
condition (4.1.2) for $\INT$ of Definition~\ref{def:Discardability} is
satisfied for any literal after $\non l$ in chains for intentions, and such
chains can be truncated at $\non l$. Furthermore, if the algorithm has
already proven $+\partial_{\OBL} l$, then the obligation of $l$ has been
violated. Thus, $l$ can be removed from all chains for obligations (line
\ref{DrigaCaseBIf+OR}). If instead $-\partial_{\OBL} l$ holds, then the
elements after $\non l$ in chains for social intentions satisfy condition
(4.1.2) of Definition~\ref{def:Discardability}, and the algorithm removes
them (line \ref{DrigaCaseBIf-OR}).

\subsection{Computational Results} 
\label{subsec:computational_results}

We now present the computational properties of the algorithms
previously described. Since
Algorithms~\ref{alg:proved}~\textsc{Proved} and
\ref{alg:refuted}~\textsc{Refuted} are sub-routines of the main one, we shall exhibit the correctness and completeness results of these algorithms inside
theorems for Algorithm~\ref{alg:defeasible}~\textsc{DefeasibleExtension}. In order to properly demonstrate results on the complexity of the algorithms, we
need the following definition.

\begin{definition}[Size of a theory]\label{def:TheorySize}
	Given a finite defeasible theory $D$, the \emph{size} $S$ of $D$ is
the number of occurrences of literals plus the number of the rules in $D$.
\end{definition}
For instance, the size of the theory
\begin{align*}
  F = \{a,\ \OBL b\}\quad
  R = \{r_1: a \To_{\OBL} c,\ r_{2}: a, \OBL b \To d\} 
\end{align*}
is equal to nine, since literal $a$ occurs three times.

We also report some key ideas and intuitions behind our implementation.
\begin{enumerate}
\item Each operation on global sets $\pm\partial_{\blacksquare}$ and
$\partial^{\pm}_{\blacksquare}$ requires linear time, as we manipulate finite sets of literals;

\item For each literal $\Box l \in HB$, we implement a hash table with
pointers to the rules where the literal occurs in; thus, retrieving the set
of rules containing a given literal requires constant time;

\item The superiority relation can also be implemented by means of hash
tables; once again, the information required to modify a given tuple can be
accessed in constant time.
\end{enumerate}
In Section~\ref{sec:algorithmic_results} we discussed the main intuitions 
behind the operations performed by the algorithms, and we explained
that each operation corresponds to a reduction that transforms a theory in
an equivalent smaller theory. Appendix~\ref{sec:CorrDefExt} exhibits
a series of lemmas stating the conditions under which an operation that
removes either rules or literals form either the head or rules or from the body 
results in an equivalent smaller theory. The Lemmas proved by induction on 
the length of derivations.  

\begin{restatable}{theorem}{Termination}\label{lem:ComplexityMMDisc}
	Given a finite defeasible theory $D$ with size $S$, Algorithms~\ref{alg:proved}~\textsc{Proved} and
	\ref{alg:refuted}~\textsc{Refuted} terminate and their computational complexity is $O(S)$.
\end{restatable}

\begin{restatable}{theorem}{Complexity}\label{thm:ComplexityCD}
	Given a finite defeasible theory $D$ with size $S$,
Algorithm~\ref{alg:defeasible}~\textsc{DefeasibleExtension} terminates and its
computational complexity is $O(S)$.
\end{restatable}

\begin{restatable}{theorem}{Correctness}\label{thm:SoundCompl}
	Algorithm~\ref{alg:defeasible}~\textsc{DefeasibleExtension} is sound and
complete.
\end{restatable}

\section{Summary and Related Work} 
\label{sec:related_work}

This article provided a new proposal for extending DL to model cognitive
agents interacting with obligations. We distinguished concepts of desire,
goal, intention and social intention, but we started from the shared notion
of outcome. Therefore, such concepts spring from a single notion that becomes distinct based
on the particular relationship with beliefs and norms. This reflects a more
natural notion of mental attitude and can express the well-known notion of
Plan B. When we consider the single chain itself, this justifies that from a
single concept of outcome we can derive all the other mental attitudes.
Otherwise we would need as many additional rules as the elements in the chain;
this, in turn, would require the introduction of additional notions to
establish the relationships with beliefs and norms. This adds to our
framework an economy of concepts.

Moreover, since the preferences allow us to determine what preferred outcomes
are adopted by an agent (in a specific scenario) when previous elements in
sequences are no longer feasible, our logic provides an abstract semantics
for several types of goal and intention reconsideration.

A drawback of our approach perhaps lies in the difficulty of translating a
natural language description into a logic formalisation. This is a
notoriously hard task. Even if the obstacle seems very difficult, the payoff
is worthwhile. The first reason is due to the efficiency of the computation
of the positive extension once the formalisation has been done (polynomial
time against the majority of the current frameworks in the literature which
typically work in exponential time). The second reason is that the use of
rules (such as business rules) to describe complex systems is extremely common
\citep{knolmayer2000modeling}. Future lines of research will then focus on
developing such methods, by giving tools which may help the (business) analyst
in writing such (business) rules from the declarative description.

The logic presented in this paper, as the vast majority of approaches to
model autonomous agents, is propositional. The
algorithms to compute the extension of theory relies on the theory being
finite, thus the first assumption for possible first-order extensions would
be to work on finite domains of individuals. Given this assumption, the
algorithms can be still be used once a theory has been grounded. This means
that the size of theory is in function of the size of the grounding. We 
expect that the size of the grounding depends on the cardinality of the 
domain of individuals and the length of the vector obtained by the join 
of the predicates occurring in the theory.

Our contribution has strong connections with those by
\cite{lpar05,jaamas:08,DBLP:journals/igpl/GovernatoriPRS09}, but it completely
rebuilds the logical treatment of agents' motivational attitudes by presenting
significant innovations in at least two respects.

First, while in \citep{lpar05,jaamas:08,DBLP:journals/igpl/GovernatoriPRS09}
the agent deliberation is simply the result of the derivation of mental
states from \emph{precisely} the corresponding rules of the logic---besides
conversions, intentions are derived using only intention rules, goals using
goal rules, etc.---here, the proof theory is much more aligned with the BDI
intuition, according to which intentions and goals are the results of the
manipulation of desires. The conceptual result of the current paper is
that this idea can be entirely encoded within a logical language and a proof
theory, by exploiting the different interaction patterns between the basic
mental states, as well as the derived ones. In this perspective, our
framework is significantly richer than the one in BOID
\citep{broersen2002goal}, which uses different rules to derive the
corresponding mental states and proposes simple criteria to solve conflicts
between rule types.

Second, the framework proposes a rich language expressing two orthogonal
concepts of preference among motivational attitudes. One is encoded within
$\odot$ sequences, which state (reparative) orders among
homogeneous mental states or motivations. 
The second type of preference is encoded via the superiority
relation between rules: the superiority can work locally between single rules
of the same or different types, or can work systematically by stating via $\Conflict(X,Y)$ that two different motivations $X$ and $Y$
collide, and $X$ always overrides $Y$. The interplay between these two
preference mechanisms can help us in isolating different and complex ways for
deriving mental states, but the resulting logical machinery is still
computationally tractable, as the algorithmic analysis proved.

Lastly, since the preferences allow us to determine what preferred outcomes
are adopted by an agent when previous elements in $\odot$-sequences are not
(or no longer) feasible, our logic in fact provides an abstract semantics for
several types of goal and intention reconsideration. Intention
reconsideration was expected to play a crucial role in the BDI paradigm
\citep{bratman1,DBLP:journals/ai/CohenL90} since intentions obey the law of
inertia and resist retraction or revision, but they can be reconsidered when
new relevant information comes in \citep{bratman1}. Despite that, the problem
of revising intentions in BDI frameworks has received little attention. A
very sophisticated exception is that of \cite{hoek}, where revisiting intentions
mainly depends on the dynamics of beliefs but the process is incorporated in
a very complex framework for reasoning about mental states. Recently,
\cite{Shapiro:2012} discussed how to revise the commitments to planned
activities because of mutually conflicting intentions, a contribution that interestingly
has connections with our work. How to employ our logic to give a semantics
for intention reconsideration is not the main goal of the paper and is left
to future work.


Our framework shares the motivation with that of
\cite{Winikoff02declarativeprocedural}, where the authors provide a logic to
describe both the declarative and procedural nature of goals. The nature of
the two approaches lead to conceptually different solutions. For instance,
they require goals, as in \citep{DBLP:conf/atal/HindriksBHM00}, ``not to be
entailed by beliefs, i.e., that they be unachieved'', while our beliefs can
be seen as ways to achieve goals. Other requirements such as persistence or
dropping a goal when reached cannot be taken into account.

\cite{DBLP:journals/logcom/ShapiroLL07} and
\cite{DBLP:conf/dagstuhl/ShapiroB07} deal with goal change. The authors
consider the case where an agent readopts goals that were previously believed
to be impossible to achieve up to revision of her beliefs. They model goals
through an accessibility relation over possible worlds. This is similar to
our framework where different worlds are different assignments to the set of
facts. Similarly to us, they prioritise goals as a preorder $\leq$; an agent
adopts a new goal unless another incompatible goal prior in the ordering
exists. This is in line with our framework where if we change the set of
facts, the algorithms compute a new extension of the theory where two
opposite literals can be proved as $\DES$ but only one as $\INT$. Notice also
that the ordering used in their work is unique and fixed at design time,
while in our framework chains of outcome rules are built trough a
context-dependent partial order which, in our opinion, models more realistic
scenarios.

\cite{DBLP:conf/atal/DastaniRM06} present three types of declarative goals: perform, achievement, and maintenance goals. In particular, they
define planning rules which relate configurations of the world as seen by the
agent (i.e., her beliefs). A planning rule is considered \emph{correct} only
if the plan associated to the rule itself allows the agent to reach a
configuration where her goal is satisfied. This is strongly connected to our
idea of belief rules, which define a path to follow in order to reach an
agent outcome. Notice that this kind of
research based on temporal aspects is orthogonal to ours.

The unifying framework proposed by
\cite{vanRiemsdijk:2008:GAS:1402298.1402323} and
\cite{DBLP:conf/atal/DastaniRW11} specifies different facets of the concept
of goal. However, several aspects make a comparative analysis between the two
frameworks unfeasible. Their analysis is indeed merely taxonomical, and it
does not address how goals are used in agent logics, as we precisely do here.

\cite{DBLP:journals/aamas/RiemsdijkDM09} share our aim to formalise goals in
a logic-based representation of conflicting goals and propose two different
semantics to represent \emph{conditional} and \emph{unconditional} goals.
Their central thesis, supported by \cite{DBLP:conf/comma/Prakken06}, is that
only by adopting a credulous interpretation is it possible to have
conflicting goals. However, we believe that a credulous interpretation is not
suitable if an agent has to deliberate what her primary goals are in a given
situation. We opted to have a sceptical interpretation of the concepts we
call goals, intentions, and social intentions, while we adopt a credulous
interpretation for desires. Moreover, they do not take into account the
distinction between goals and related motivational attitudes (as in
\citep{vanRiemsdijk:2008:GAS:1402298.1402323,
DBLP:conf/atal/DastaniRW11,DBLP:conf/atal/DastaniRM06}). The characteristic
property of intentions in these logics is that an agent may not drop
intentions for arbitrary reasons, which means that intentions have a certain
persistence. As such, their analysis results orthogonal to ours.

\cite{DBLP:journals/aamas/VasconcelosKN09} propose mechanisms for the
detection and resolution of normative conflicts. They resolve conflicts by
manipulating the constraints associated to the norms' variables, as well as through \emph{curtailment}, that is reducing the scope of the norm. In other works, we
dealt with the same problems in defeasible deontic logic
\citep{DBLP:journals/jphil/GovernatoriORS13}.
 We found three problems in their solution: (i) the curtailing relationship
$\omega$ is rather less intuitive than our preference relation $>$, (ii) their
approach seems too convoluted in solving exceptions (and they do not provide
any mechanism to handle reparative chains of obligations), and (iii) the space
complexity of their \emph{adoptNorm} algorithm is exponential.

The present framework is meant to be seen as the first step within a more
general perspective of providing the business analyst with tools that allow
the creation of a business process in a fully declarative manner
\citep{DBLP:conf/prima/OlivieriGSC13}. Another issue comes from the fact
that, typically, systems implemented by business rules involve thousands of
such rules. Again, our choice of Defeasible Logic allows to drastically
reduce the number of rules involved in the process of creating, for example, a business
process thanks to its exception handling mechanism. This is peculiarly
interesting when dealing with the problem of visualising such rules. When
dealing with a system with thousands of rules, understanding what they
represent or what a group of rules stand for, may be a serious challenge. On
the contrary, the model presented by
\cite{DBLP:conf/prima/OlivieriGSC13}, once an input is given, allows for the
identification of whether the whole process is compliant against a normative
system and a set of goals (and if not, where it fails). To the best of our
knowledge, no other system is capable of checking whether a process can start
with its input requisites and reaches its final objectives in a way that is
compliant with a given set of norms.

\paragraph{Acknowledgements}
NICTA is funded by the Australian Government through the Department of 
Communications and the Australian Research Council through the ICT Centre 
of Excellence Program.

This paper is an extended and revised version of \cite{DBLP:conf/ruleml/GovernatoriORSC13} presented at the  
7th International Symposium on Theory, Practice, and Applications of Rules on the Web (RuleML 2013).
We thank all the anonymous reviewers for their valuable comments. 

\bibliographystyle{abbrvnat}
\bibliography{bibexport}

\appendix
\section{Inferential mechanism example}\label{sec:table}

This appendix is meant to offer more details for illustrating the
inference mechanism proposed in this paper, and to consider Definition
\ref{def:proofCond+X} more carefully. Therefore, first we report in Table
\ref{tab:Condition (2.3.2)} the most interesting scenarios, where a rule $r$
proves $+\partial_{\INTS}q$ when attacked by an applicable rule $s$, which in
turn is successfully counterattacked by an applicable rule $t$. Lastly, we
end this appendix by reporting an example. The situation described there
starts from a natural language description and then shows how it can be
formalised with the logic we proposed.

For the sake of clarity, notation
$\BEL, \Box$ (with $\Box \in \set{\OBL, \INTS}$) represents belief rules
which are Conv-applicable for mode $\Box$. 

For instance, the sixth row of the table denotes situations like the
following:

\begin{align*}
	F & = \{ a,\, b,\, \OBL c \} \\
	R & = \{ r: a \To_{\OUT} q,\\
	  & \psl     s: b \To_{\OBL} \neg q,\\
	  & \psl     t: c \To_{} q\}\\
	> & = \{ (t, s) \}.\\
\end{align*}

The outcome rule $r$ for $q$ is applicable for $\INTS$ according to Definition
\ref{def:Applicability}. Since in our framework we have $\Aconf{\OBL}{\INTS}$,
the rule $s$ for $\neg q$ (which is applicable for $\OBL$) does not satisfy
condition (2.3.1) of Definition \ref{def:proofCond+X}. As a result, $s$
represents a valid attack to $r$. However, since we have $\Aconv{\BEL}{\OBL}$, rule $t$ is Conv-applicable for $\OBL$ by Definition
\ref{def:Conv-appl}, with $t>s$ by construction. Thus, $t$ satisfies condition (2.3.2.1) of
Definition \ref{def:proofCond+X} and successfully
counterattacks $s$. Consequently, $r$ is able to conclude $+\partial_{\INTS}q$.


\begin{table}[h]	
\centering
\begin{tabular}{cccc}
	\toprule
 \textbf{Mode of $r$} & \textbf{Mode of $s$} & \textbf{Mode of $t$} & \textbf{$+\partial_{\INTS}q$ because$\ldots$} \\

\hline

\OUT applicable for \INTS & \OUT applicable for \INTS & \OUT applicable for \INTS & $t > s$\\
\OUT applicable for \INTS & \OUT applicable for \INTS & \OBL & $\Conflict(\OBL,\INTS)$ \\
\OUT applicable for \INTS & \OUT applicable for \INTS & $\BEL,\INTS$ & $t > s$\\
\OUT applicable for \INTS & \OUT applicable for \INTS & $\BEL,\OBL$ & $\Conflict(\OBL,\INTS)$ \\
\OUT applicable for \INTS & \OBL & \OBL & $t > s$\\
\OUT applicable for \INTS & \OBL & $\BEL,\OBL$ & $t > s$\\
\OUT applicable for \INTS & $\BEL,\INTS$ & \OUT applicable for \INTS & $t > s$\\
\OUT applicable for \INTS & $\BEL,\INTS$ & \OBL & $\Conflict(\OBL,\INTS)$\\
\OUT applicable for \INTS & $\BEL,\INTS$ & $\BEL,\INTS$ & $t > s$\\
\OUT applicable for \INTS & $\BEL,\INTS$ & $\BEL,\OBL$ & $\Conflict(\OBL,\INTS)$\\
\OUT applicable for \INTS & $\BEL,\OBL$ & $\BEL,\OBL$ & $t > s$\\
$\BEL,\INTS$ & \OUT applicable for \INTS & \OUT applicable for \INTS & $t > s$\\
$\BEL,\INTS$ & \OUT applicable for \INTS & \OBL & $\Conflict(\OBL,\INTS)$ \\
$\BEL,\INTS$ & \OUT applicable for \INTS & $\BEL,\INTS$ & $t > s$\\
$\BEL,\INTS$ & \OUT applicable for \INTS & $\BEL,\OBL$ & $\Conflict(\OBL,\INTS)$ \\
$\BEL,\INTS$ & \OBL & \OBL & $t > s$\\
$\BEL,\INTS$ & \OBL & $\BEL,\OBL$ & $t > s$\\
$\BEL,\INTS$ & $\BEL,\INTS$ & \OUT applicable for \INTS & $t > s$\\
$\BEL,\INTS$ & $\BEL,\INTS$ & \OBL & $\Conflict(\OBL,\INTS)$\\
$\BEL,\INTS$ & $\BEL,\INTS$ & $\BEL,\INTS$ & $t > s$\\
$\BEL,\INTS$ & $\BEL,\INTS$ & $\BEL,\OBL$ & $\Conflict(\OBL,\INTS)$\\
$\BEL,\INTS$ & $\BEL,\OBL$ & $\BEL,\OBL$ & $t > s$\\
\bottomrule
\end{tabular}
\caption{Definition \ref{def:proofCond+X}: Attacks and counterattacks for social intention}
\label{tab:Condition (2.3.2)}
\end{table}

\begin{example}\label{ex:runningEx1} 
	
\emph{PeoplEyes} is an eyeglasses manufacturer. Naturally, its final goal is
to produce cool and perfectly assembled eyeglasses. The final steps of the
production process are to shape the lenses to glasses, and mount them on the
frames. To shape the lenses, \emph{PeoplEyes} uses a very innovative and
expensive laser machine, while for the final mounting phase two different
machines can be used. Although both machines work well, the first and
newer one is more precise and faster than the other one;
\emph{PeoplEyes} thus prefers to use the first machine as much as
possible. Unfortunately, a new norm comes in force stating that no laser
technology can be used, unless human staff wears laser-protective goggles.
\end{example}

If \emph{PeoplEyes} has both human resources and raw material, and the three
machines are fully working, but it has not yet bought any laser-protective
goggles, all its goals would be achieved but it would fail to comply with the
applicable regulations, since the norm for the no-usage of laser technology is violated
and not compensated.

If \emph{PeoplEyes} buys the laser-protective goggles, their entire production
process also becomes norm compliant. If, at some time, the more precise
mounting machine breaks, but the second one is still working,
\emph{PeoplEyes} can still reach some of its objectives since the usage of the
second machine leads to a state of the world where the objective of mounting
the glasses on the frames is accomplished. Again, if \emph{PeoplEyes} has no
protective laser goggles and both the mounting machines are out of order,
\emph{PeoplEyes}' production process is neither norm, nor outcome compliant.

The following theory is the formalisation into our logic of the above scenario.
\allowdisplaybreaks
\begin{align*}
    F & =\{ \mathit{lenses},\; \mathit{frames},\; \mathit{new\_safety\_regulation}\} \\
    R & =\{r_{1}: \To_{\OUT} \mathit{eye\_Glasses}\\
	& \psl r_{2}: \To \mathit{laser}\\
	& \psl r_{3}: \mathit{lenses}, \mathit{laser} \To \mathit{glasses}\\
	& \psl r_{4}: \To \mathit{mounting\_machine1}\\
	& \psl r_{5}: \To \mathit{mounting\_machine2}\\
	& \psl r_{6}: \mathit{mounting\_mach1} \To \neg\mathit{mounting\_machine2}\\
	& \psl r_{7}: \mathit{frames}, \mathit{glasses}, \mathit{mounting\_machine1} \To\mathit{eye\_Glasses}\\
    & \psl r_{8}: \mathit{frames}, \mathit{glasses}, \mathit{mounting\_machine2} \To \mathit{eye\_Glasses}\\
	& \psl r_{9}: \mathit{new\_safety\_regulation} \To_{\OBL} \neg\mathit{laser} \otimes \mathit{goggles}\\
    & \psl r_{10}: \To_{\OUT} \mathit{mounting\_machine1} \oplus \mathit{mounting\_machine2}
\}\\
    >^{sm} & =\{r_{6}>r_{5}\}.
\end{align*}
%
We assume \textit{PeoplEye} has enough resources to start the process by
setting $\mathit{lenses}$ and $\mathit{frames}$ as facts. Rule $r_1$ states
that producing $\mathit{eye\_Glasses}$ is the main objective ($+\partial_{\INT}
\mathit{eye\_Glasses}$, we choose intention as the mental attitude to comply
with/attain to); rules $r_2$, $r_4$ and $r_5$ describe that we can use,
respectively, the laser and the two mounting machineries. Rule $r_3$ is to
represent that, if we have lenses and a laser machinery available, then we
can shape glasses; in the same way, rules $r_7$ and $r_8$ describe that
whenever we have glasses and one of the mounting machinery is available, then
we obtain the final product. Therefore, the positive extension for belief
$+\partial$ contains $\mathit{laser}$, $\mathit{glasses}$,
$\mathit{mounting\_machine1}$ and $\mathit{eye\_Glasses}$. In that occasion,
rule $r_6$ along with $>$ prevent the using of both machineries at the same
time and thus $-\partial \mathit{mounting\_machine2}$ (we assumed, for
illustrative purpose even if unrealistically, that a parallel execution is
not possible). When a new safety regulation comes in force ($r_9$), the usage
of the laser machinery is forbidden, unless protective goggles are worn
($+\partial_{\OBL}\neg \mathit{laser}$ and $+\partial_{\OBL}\neg
\mathit{goggles}$). Finally, rule $r_{10}$ is to describe the preference of
using $\mathit{mounting\_machine1}$ instead of $\mathit{mounting\_machine2}$
(hence we have $+\partial_{\INT} \mathit{mounting\_machine1}$ and
$-\partial_{\INT} \mathit{mounting\_machine2}$).

Since there exists no rule for goggles, the theory is outcome compliant (that is, it reaches some set of objectives), but
not norm compliant (given that it fails to meet some obligation rules without compensating them). If we add $\mathit{goggles}$ to the facts and we substitute $r_2$ with
\[ 
r'_{2}: \OBL \mathit{goggles} \Rightarrow \mathit{laser} 
\]
then we are both norm and outcome compliant, as well as if we add \[ r_{11}:
\mathit{mounting\_machine1\_broken} \Rightarrow \neg\mathit{mounting\_machine1} \] to $R$ and
$\mathit{mounting\_machine1\_broken}$ to $F$. Notice that, with respect to
$laser$, we are intention compliant but \textit{not} social intention
compliant (given $\OBL \neg \mathit{lenses}$). This is a key characteristic
of our logic: The system is informed that the process is compliant but some
violations have occurred.

\section{Proofs of Propositions in Section~\ref{sec:logic}}\label{sec:proofs}
\begingroup
\def\Box{X}
\Coherence*
\begin{proof}
 \emph{1. (Coherence of the logic)} The negative proof tags are the strong negation of the positive
ones, and so are the conditions of a rule being discarded
(Definition~\ref{def:Discardability}) for a rule being applicable
(Definition~\ref{def:Applicability}). Hence, when the conditions for $+\partial_{\Box}$ hold, those for $-\partial_{\Box}$ do not.

 \emph{2. (Consistency of the logic)} We split the proof into two cases: (i)
at least one of $\Box l$ and $\Box \non l$ is in \FACTS, and (ii) neither of them
is in \FACTS. For (i) the proposition immediately follows by the assumption of
consistency. In fact, suppose that $\Box l \in \FACTS$. Then clause (1) of
$+\partial_{\Box}$ holds for $l$. By consistency $\Box \non l\not\in \FACTS$,
thus clause (1) of Definition~\ref{def:proofCond+X} does not hold for $\non l$.
Since $\Box l \in \FACTS$, also clause (2.1) is always falsified for $\non l$,
and the thesis is proved.

For (ii), let us assume that both $+\partial_{\Box}l$ and
$+\partial_{\Box}\non l$ hold in $D$. A straightforward assumption derived by
Definitions~\ref{def:Applicability} and~\ref{def:Discardability} is that no
rule can be at the same time applicable and discarded for $\Box$ and $l$ for
any literal $l$ and its complement. Thus, we have that there are applicable
rules for $\Box$ and $l$, as well as for $\Box$ and $\non l$. This means that
clause (2.3.2) of Definition~\ref{def:proofCond+X} holds for both $l$ and
$\non l$. Therefore, for every applicable rule for $l$ there is an applicable
rule for $\non l$ stronger than the rule for $l$. Symmetrically, for every
applicable rule for $\non l$ there is an applicable rule for $l$ stronger
than the rule for $\non l$. Since the set of rules in $D$ is finite by
construction, this situation is possible only if there is a cycle in the
transitive closure of the superiority relation, which is in contradiction
with the hypothesis of $D$ being consistent.
\end{proof}

\PositiveProp*
\begin{proof}
For part 1., let $D$ be a consistent defeasible theory, and
$D\vdash+\partial_{\Box}l$. Literal $\non l$ can be in only one of the
following, mutually exclusive situations: (i) $D\vdash +\partial_{\Box} \non
l$; (ii) $D\vdash -\partial_{\Box} \non l$; (iii) $D\not\vdash
\pm\partial_{\Box} \non l$. Part 2 of Proposition~\ref{prop:CoherenceConsistence} allows us to exclude case (i), since $D\vdash +\partial_{\Box} l$
by hypothesis. Case (iii) denotes situations where there are loops in the
theory involving literal $\non l$,\footnote{For example, situations like $\Box
\non l \To_{\Box} \non l$, where the proof conditions generate a loop
without introducing a proof.} but inevitably this would affect also the
provability of $\Box l$, i.e., we would not be able to give a proof for
$+\partial_{\Box} l$ as well. This is in contradiction with the hypothesis.
Consequently, situation (ii) must be the case.

Parts 2. and 3. directly follow by
Definitions~\ref{def:Applicability} and \ref{def:Discardability}, while Definitions~\ref{def:Applicability} and
\ref{def:proofCond+X} justify part 4., given that $\GOAL$ is not involved in any conflict relation.

Part 5. Trivially, from part 4.
\end{proof}

\NegativeProp*
\begin{proof}
Example~\ref{ex:Jsick} in the extended version offers counterexamples showing
the reason why the above statements do not hold.
\begin{align*}
 F  & = \{\mathit{saturday},\, \mathit{John\_away},\, \mathit{John\_sick} \} \\
 R  & = \{r_2: \mathit{saturday} \Rightarrow_{\OUT} \mathit{visit\_John}\odot
  \mathit{visit\_parents} \odot \mathit{watch\_movie}\\
  & \psl r_3: \mathit{John\_away} \Rightarrow_{\Bel} \neg \mathit{visit\_John}\\
  & \psl r_{4}: \mathit{John\_sick} \To_{\OUT} \neg \mathit{visit\_John}\odot \mathit{short\_visit}\}\\
  & \psl r_{7}: \mathit{John\_away} \To_{\BEL} \neg \mathit{short\_visit}\}\\
 >  & = \{(r_{2}, r_{4}) \}.
\end{align*}
Given that $r_{2}> r_{4}$, Alice has the desire to $\mathit{visit\_John}$, and this is
also her preferred outcome. Nonetheless, being $\mathit{John\_away}$ a fact, this is not
her intention, while so are $\neg \mathit{visit\_John}$ and $\mathit{visit\_parents}$.
\end{proof}

\endgroup
\section{Correctness and Completeness of \textsc{DefeasibleExtension}}\label{sec:CorrDefExt}

In this appendix we give proofs of the lemmas used by
Theorem~\ref{thm:SoundCompl} for the soundness and completeness of the
algorithms proposed.

We recall that the algorithms in Section~\ref{sec:algorithmic_results} are based on a series
of transformations that reduce a given theory into an equivalent, simpler
one. Here, equivalent means that the two theories have the same extension,
and simpler means that the size of the target theory is smaller than that of
the original one. Remember that the size of a theory is the number of
instances of literals occurring in the theory plus the number of rules in the theory. Accordingly, each
transformation either removes some rules or some literals from rules
(specifically, rules or literals we know are no longer useful to produce new
conclusions). There is an exception. At the beginning of the computation, the
algorithm creates four rules (one for each type of goal-like attitude) for
each outcome rule (and the outcome rule is then eliminated). The purpose of
this operation is to simplify the transformation operations and the
bookkeeping of which rules have been used and which rules are still able to
produce new conclusions (and the type of conclusions). Alternatively, one
could implement flags to achieve the same result, but in a more convoluted
way. A consequence of this operation is that we no longer have outcome rules.
This implies that we have (i) to adjust the proof theory, and (ii) to show that the
adjusted proof theory and the theory with the various goal-like rules are
equivalent to the original theory and original proof conditions.

The adjustment required to handle the replacement of each outcome rule with a
set of rules of goal-like modes (where each new rule has the same body and
consequent of the outcome rule it replaces) is to modify the definition of
being applicable (Definition~\ref{def:Applicability}) and being discarded
(Definition~\ref{def:Discardability}). Specifically, we have to replace 
\begin{itemize}
  \item $r\in R^{\OUT}$ in clause 3 of Definition~\ref{def:Applicability} with
    $r\in R^\DES$;
  \item $r\notin R^{\OUT}$ in clause 3 of Definition~\ref{def:Discardability}
    with $r\notin R^\DES$;
  \item $r\in R^{\OUT}$ in clause 4.1.1 of Definition~\ref{def:Applicability}
    with $r\in R^X$; and
  \item $r\notin R^{\OUT}$ in clause 4.1.1 of Definition~\ref{def:Discardability}
    with $r\notin R^X$.
\end{itemize}
Given a theory $D$ with goal-like rules instead of outcome rules we will use 
$E_3(D)$ to refer to the extension of $D$ computed using the proof theory obtained 
from the proof theory defined in Section~\ref{sec:logic} with the modified versions 
of the notions of applicable and discarded just given.
\begin{lem}\label{lem:addRandSup}
Let $D=(F,R,>)$ be a defeasible theory. Let $D'=(F, R',>')$ be the defeasible theory obtained from $D$ as follows:
\begin{align*}
 R' = {}& R^{\BEL} \cup R^{\OBL} \cup \set{r_X\colon A(r) \hookrightarrow_X C(r)| r\colon A(r)\hookrightarrow_{\OUT} C(r)\in R, X\in\set{\DES, \GOAL, \INT, \INTS}} \\
 >' = {}&  \set{(r,s)| (r,s)\in {>},\ s,r\in R^\BEL \cup R^\OBL}\cup
           \set{(r_X,s_Y) | (r,s) \in {>}, r,s\in R^\OUT} \cup {}\\
    &  \set{(r_X,s)| (r,s)\in {>}, r\in R^\OUT, s\in R^\BEL \cup R^\OBL} \cup
      \set{(r,s_X)| (r,s)\in {>}, r\in R^\BEL\cup R^\OBL, s\in R^\OUT}
\end{align*}
Then, $E(D) = E_3(D')$.
\end{lem}
\begin{proof}
The differences between $D$ and $D'$ are that each outcome-rule in $D$ corresponds
to four rules in $D'$ each for a different mode and all with the same antecedent 
and consequent of the rule in $D$. Moreover, every time a rule  $r$ in $D$ is 
stronger than a rule $s$ in $D$, then any rule corresponding to $r$ in $D'$ 
is stronger than any rule corresponding to $s$ in $D'$.

The differences in the proof theory for $D$ and that for $D'$ is in the definitions 
of applicable for $X$ and  discarded for $X$. It is immediate to verify that 
every time a rule $r$ is applicable (at index $n$) for $X$, then $r_X$ is applicable 
(at index $n$) for $X$ (and the other way around). 
\end{proof}

Given the functional nature of the transformations involved in the algorithms, we shall refer to the
rules in the target theory with the same labels as the rules in the
source theory. Thus, given a rule $r\in D$, we will refer to the rule
corresponding to it in $D'$ (if it exists) with the same label, namely $r$.

In the algorithms, belief rules may convert to another mode $\Diamond$ only
through the support set $R^{\BEL, \Diamond}$. Definition~\ref{def:Conv-appl}
requires $R^{\BEL, \Diamond}$ to be initialised with a modal version of
each belief rule with \emph{non-empty} antecedent, such that every literal $a$
in the antecedent is replaced by the corresponding modal literal $\Diamond a$. 

In this manner, rules in $R^{\BEL, \Diamond}$ satisfy clauses 1 and 2 of
Definitions~\ref{def:Conv-appl} and \ref{def:Conv-disc} by construction, while
clauses 3 of both definitions are satisfied iff these new rules for
$\Diamond$ are body-applicable (resp. body-discarded). Therefore, conditions
for rules in $R^{\BEL, \Diamond}$ to be applicable/discarded collapse into
those of Definition~\ref{def:BodyAppl} and \ref{def:BodyDisc}, and accordingly
these rules are applicable for mode $\Diamond$ only if they satisfy clauses clauses (2.1.1), (3.1), or (4.1.1) of
Definitions~\ref{def:Applicability} and \ref{def:Discardability}, based on how $\Diamond$ is instantiated. That is to say, during the execution of the
algorithms, we can empty the body of the rules in $R^{\BEL, \Diamond}$ by
iteratively proving all the modal literals in the antecedent to decide which
rules are applicable at a given step.

Before proceeding with the demonstrations of the lemmas, we recall that in the
formalisation of the logic in Section \ref{sec:logic}, we referred to modes with
capital roman letters ($X$, $Y$, $T$) while the notation of the algorithms in Section 
\ref{sec:algorithmic_results}
proposes the variant with $\Box$, $\blacksquare$ and $\Diamond$ since it was
needed to fix a given modality for the iterations and pass the correct input
for each call of a subroutine. Therefore, being that the hypotheses of the lemmas
refer to the operations performed by the algorithms, while the proofs refer
to the notation of Definitions \ref{def:BodyAppl}--\ref{def:complement}, in the following the former ones use
the symbol $\Box$ for a mode, the latter ones the capital roman letters notation.

\begin{lem}\label{lem:+PartialBox}
	Let $D=(\FACTS,R,>)$ be a defeasible theory such that $D\vdash
+\partial_{\Box} l$ and $D'=(\FACTS,R',>')$ be the theory obtained from $D$
where
\begin{align*}
	R' = & \set{r: A(r) \setminus \set{\Box l, \neg \Box \non l}\hookrightarrow C(r) |\ r\in R,\ A(r) \cap \widetilde{\Box l} = \emptyset} \\
	R'^{\BEL,\Box} = & \set{r: A(r)\setminus \set{\Box l}\hookrightarrow C(r)| r\in R^{\BEL,\Box},\ A(r) \cap \widetilde{\Box l} = \emptyset}\\
    >' = & > \setminus \set{(r,s), (s,r) \in {>} |\  A(r) \cap \widetilde{\Box l} \not = \emptyset}.
\end{align*}
Then $D\equiv D'$.	
\end{lem}
\begin{proof}
	The proof is by induction on the length of a derivation $P$. For the inductive base, we consider all possible derivations for a literal $q$ in the theory.

\paragraph{$P(1)= +\partial_{X} q$, with $X \in \MOD \setminus \set{\DES}$.}{
This is possible in two cases: (1) $X q\in \FACTS$, or (2) $\widetilde{Y q}
\cap \FACTS = \emptyset$, for $Y=X$ or $\Conflict(Y,X)$, and $\exists r\in
R^{X}[q,i]$ that is applicable in $D$ for $X$ at $i$ and $P(1)$, and every
rule $s\in R^{Y}[\non q,j]$ is either (a) discarded for $X$ at $j$ and $P(1)$,
or (b) defeated by a stronger rule $t\in R^{T}[q,k]$ applicable for $T$ at $k$
and $P(1)$ ($T$ may conflict with $Y$).

Concerning (1), by construction of $D'$, $X q\in \FACTS$ iff $X q\in \FACTS'$,
thus if $+\partial_{X} q$ is provable in $D$ then is provable in $D'$, and
vice versa.

Regarding (2), again by construction of $D'$, $\widetilde{Y q} \cap
\FACTS=\emptyset$ iff $\widetilde{Y q} \cap \FACTS'=\emptyset$. Moreover, $r$ is applicable at
$P(1)$ iff $i=1$ (since lemma's operations do not modify the tail of the
rules) and $A(r) = \emptyset$. Therefore, if $A(r) = \emptyset$ in $D$ then
$A(r)=\emptyset$ in $D'$. This means that if a rule is applicable in $D$ at
$P(1)$ then is applicable in $D'$ at $P(1)$. In the other direction, if $r$ is
applicable in $D'$ at $P(1)$, then either (i) $A(r) = \emptyset$ in $D$, or
(ii) $A(r) = \set{\Box l}$, or $A(r) = \set{\neg \Box \non l}$. For (i), $r$
is straightforwardly applicable in $D$, as well as for (ii) since $D\vdash
+\partial_{\Box} l$ by hypothesis.

When we consider possible attacks to rule $r$, namely $s\in R^{Y}[\non q, j]$,
we have to analyse cases (a) and (b) above.

(a) Since we reason about $P(1)$, it must be the case that no such rule $s$
exists in $R$, and thus $s$ cannot be in $R'$ either.
In the other direction, the difference between $D$ and $D'$ is that in $R$ we
have rules with $\widetilde{\Box l}$ in the antecedent, and such rules are not
in $R'$. Since $D\vdash +\partial_{\Box} l$ by hypothesis, all rules in $R$
for which there is no counterpart in $R'$ are discarded in $D$.

(b) We modify the superiority relation by only withdrawing instances where one
of the rules is discarded in $D$. But only when $t$ is applicable then is
active in the clauses of the proof conditions where the superiority relation
is involved, i.e., (2.3.2) of Definition~\ref{def:proofCond+X}. We have just
proved that if a rule is applicable in $D$ then is applicable in $D'$ as well,
and if is discarded in $D$ then is discarded in $D'$. If $s$ is not discarded
in $D$ for $Y$ at $1$ and $P(1)$, then there exists an applicable rule $t$ in
$D$ for $q$ stronger than $s$. Therefore $t$ is applicable in $D'$ for $T$ and
$t >' s$ if $T=Y$, or $\Conflict(T,Y)$. Accordingly, $D'\vdash +\partial_{X}
q$. The same reasoning applies in the other direction. Consequently, if we
have a derivation of length 1 of $+\partial q$ in $D'$, then we have a
derivation of length 1 of $+\partial q$ in $D$ as well.

Notice that in the inductive base by their own nature rules in $R^{\BEL,
\Diamond}$, even if can be modified or erased, cannot be used in a proof of
length one.}

\paragraph{$P(1)= +\partial_{\DES} q$.}{
The proof is essentially identical to the inductive base for $+\partial_{X}
q$, with some slight modifications dictated by the different proof
conditions for $+\partial_{\DES}$: (1) $\DES q\in \FACTS$, or (2) $\neg \DES q
\not\in \FACTS$, and $\exists r\in R^{\DES}[q,i]$ that is applicable for
$\DES$ at 1 and $P(1)$ and every rule $s\in R^{\DES}[\non q,j]$ is either (a)
discarded for $\DES$ at 1 and $P(1)$, or (b) $s$ is not stronger than $r$.

\paragraph{$P(1)= -\partial_{X} q$ with $X\in \MOD$.}{Clearly conditions (1)
and (2.1) of Definition~\ref{def:proofCond-X} hold in $D$ iff they do in $D'$,
given that $\FACTS=\FACTS'$. The analysis for clause (2.2) is the same of case
(a) of $P(1)=+\partial_{X} q$, while for clause (2.3.1) the reader is referred
to case (2), where in both cases $r$ and $s$ change their role. For condition
(2.3.2) if $X=\DES$, then $s>r$. Otherwise, either there is no $t\in
R^{T}[q,k]$ in $D$ (we recall that at $P(1)$, $t$ cannot be discarded in $D$
because that would imply a previous step in the proof), or $t \not > s$ and
not $\Conflict(T,Y)$. Therefore $s\in R'$ by construction, and conditions on
the superiority relation between $s$ and $t$ are preserved. Hence, $D'\vdash
-\partial_{X} q$. For the other direction, we have to consider the case of a
rule $s$ in $R$ but not in $R'$. As we have proved above, all rules discarded
in $D'$ are discarded in $D$, and all rules in $R$ for which there is no
corresponding rule in $R'$ are discarded in $D$ as well, and we can process this case with the same reasoning as above.}

\medskip
	
\noindent For the inductive step, the property equivalence between $D$ and
$D'$ is assumed up to the $n$-th step of a generic proof for a given literal
$p$.}

\paragraph{$P(n+1) = +\partial_{X} q$, with $X \in \MOD$.}{
Clauses (1) and (2.1) follow the same conditions treated in the inductive base
for $+\partial_{X} q$. As regards clause (2.2), we distinguish if $X=\BEL$, or
not. In the former case, if there exists a rule $r\in R[q,i]$ applicable for
$\BEL$ in $D$, then clauses 1.--3. of Definition~\ref{def:BodyAppl} are all
satisfied. By inductive hypothesis, we conclude that the clauses are satisfied
by $r$ in $D'$ as well no matter whether $\Box l \in A(r)$, or not.

Otherwise, there exists a rule $r$ applicable in $D$ for $X$ at $P(n+1)$ such
that $r$ is either in $R^{X}[q,i]$, or $R^{\BEL, X}[q,1]$. By inductive
hypothesis, we can conclude that: (i) if $r\in R^{X}[q,i]$ then $r$ is
body-applicable and the clauses of Definition~\ref{def:BodyAppl} are satisfied
by $r$ in $D'$ as well; (ii) if $r\in R^{\BEL, X}[q,1]$ then $r$ is
Conv-applicable and the clauses of Definition~\ref{def:Conv-appl} are
satisfied by $r$ in $D'$ as well. As regards conditions (2.1.2) or (4.1.2),
the provability/refutability of the elements in the chain prior to $q$ is
given by inductive hypothesis. The direction from rule applicability in $D'$
to rule applicability in $D$ follows the same reasoning and so is
straightforward.

Condition (2.3.1) states that every rule $s\in R^{Y}[\non q, j]\cup
R^{\BEL,Y}[\non q,1]$ is discarded in $D$ for $X$ at $P(n+1)$. This means that
there exists an $a \in A(s)$ satisfying one of the clauses of
Definition~\ref{def:BodyDisc} if $s\in R^{\BEL,Y}[\non q,1]$, or
Definition~\ref{def:Discardability} if $s\in R^{Y}[\non q, j]$. Two possible
situations arise. If $a \in \widetilde{\Box l}$, then $s \notin R'$;
otherwise, by inductive hypothesis, either $a$ satisfies
Definition~\ref{def:BodyDisc} or \ref{def:Conv-disc} in $D'$ depending on
$s\in R^{Y}[\non q, j]$ or $s\in R^{\BEL,Y}[\non q,1]$. Hence, $s$ is
discarded in $D'$ as well. The same reasoning applies for the other direction.
The difference between $D$ and $D'$ is that in $R$ we have rules with elements
of $\widetilde{\Box l}$ in the antecedent, and these rules are not in $R'$.
Consequently, if $s$ is discarded in $D'$, then is discarded in $D$ and all
rules in $R$ for which there is no corresponding rule in $R'$ are discarded in
$D$ since $D\vdash +\partial_{\Box} l$ by hypothesis.

If $X\neq \DES$, then condition (2.3.2) can be treated as case (b) of the
corresponding inductive base except clause (2.3.2.1) where if $t>s$ then
either: (i) $Y=T$, (ii) $s\in R^{\BEL,T}[\non q]$ and $t\in R^{T}[q]$
($\Convert(Y,T)$), or (iii) $s\in R^{Y}[\non q]$ and $t\in R^{\BEL, Y}[q]$
($\Convert(T,Y)$). Instead if $X=\DES$, no modifications are needed.}

\paragraph{$P(n+1) = -\partial_{X} q$, with $X \in \MOD$.}{The analysis is a combination of the inductive base for $-\partial_{X} q$ and inductive step for $+\partial_{X} q$ where we have
already proved that a rule is applicable (discarded) in $D$ iff is so in $D'$
(or it is not contained in $R'$). Even condition (2.3.2.1) is just the strong
negation of the reason in the above paragraph.}
\end{proof}

\begin{lem}\label{lem:-PartialBox}
	Let $D=(\FACTS,R,>)$ be a defeasible theory such that $D\vdash -\partial_{\Box} l$ and $D'=(\FACTS,R',>')$  be the theory obtained from $D$ where
\begin{align*}
	R' = & \set{r: A(r)\setminus\set{\neg\Box l}\hookrightarrow C(r)|\ r\in R,\ \Box l \not\in A(r)} \\
	R'^{\BEL,\Box} = & \set{r\in R^{\BEL,\Box}|\ \Box l \not\in A(r)}\\
	>' = & > \setminus \set{(r,s),(s,r)\in {>} |\ \Box l\in A(r)}.
\end{align*}
Then $D\equiv D'$.	
\end{lem}

\begin{proof}
	We split the proof in two cases, depending on if $\Box\neq\DES$, or
$\Box=\DES$.

As regards the former case, since Proposition~\ref{prop:+PartialTHEN-Partial}
states that $+\partial_{X} m$ implies $-\partial_{X} \non m$ then
modifications on $R'$, $R'^{\BEL,\Box}$, and $>'$ represent a particular
case of Lemma~\ref{lem:+PartialBox} where $m = \non l$.

We now analyse the case when $\Box=\DES$. The analysis is identical to the one
shown for the inductive base of Lemma~\ref{lem:+PartialBox} but for what
follows.

\paragraph{$P(1)=+\partial_{X} q$.}{Case (2)--(ii): $A(r)=\set{\neg \Box l}$
and since $D\vdash -\partial_{\Box} l$ by hypothesis, then if $r$ is
applicable in $D'$ at $P(1)$ then is applicable in $D$ at $P(1)$ as well.

Case (2)--(a): the difference between $D$ and $D'$ is that in $R$ we have
rules with $\Box l$ in the antecedent, and such rules are not in $R'$. Since
$D\vdash -\partial_{\Box} l$ by hypothesis, all rules in $R$ for which there
is no counterpart in $R'$ are discarded in $D$. 

The same modification happens in the inductive step $P(n+1)=+\partial_{X} q$,
where also the sentence `If $a \in \widetilde{\Box l}$, then $s\notin R'$'
becomes `If $a = \Box l$, then $s\notin R'$'.

Finally, the inductive base and inductive step for the negative proof tags are
identical to ones of the previous lemma.}
\end{proof}

Hereafter we consider theories obtained by the transformations of
Lemma~\ref{lem:+PartialBox}. This means that all applicable rules are such
because their antecedents are empty and every rule in $R$ appears also in $R'$
and vice versa, and there are no modifications in the antecedent of rules.

\begin{lem}\label{lem:+PartialBEL}
	Let $D=(\FACTS,R,>)$ be a defeasible theory such that $D\vdash +\partial l$ and $D'=(\FACTS,R',>)$  be the theory obtained from $D$ where
\begin{align}
	R'^{\OBL} = & \set{ A(r) \To_{\OBL} C(r)! l |\ r\in R^{\OBL}[l,n]}\\
	R'^{\INT} = & \set{ A(r) \To_{\INT} C(r)! l |\ r\in R^{\INT}[l,n]}\  \cup\nonumber\\
			    & \set{ A(r) \To_{\INT} C(r) \ominus \non l |\ r\in R^{\INT}[\non l,n]}\\
	R'^{\INTS} = & \set{ A(r) \To_{\INTS} C(r) \ominus \non l |\ r\in R^{\INTS}[\non l,n]}.
\end{align}
Moreover, 
\begin{itemize}
	\item if $D \vdash +\partial_{\OBL}\non l$, then instead of (C1)
		\begin{align}\addtocounter{equation}{-3}
			R'^{\OBL} = & \set{ A(r) \To_{\OBL} C(r)! l |\ r\in R^{\OBL}[l,n]}\  \cup\nonumber\\
				    	& \set{ A(r) \To_{\OBL} C(r) \ominus \non l |\ r\in R^{\OBL}[\non l,n]}.
		\end{align}
	\item if $D \vdash -\partial_{\OBL}\non l$, then instead of (C3) 
		\begin{align}\addtocounter{equation}{1}
			R'^{\INTS} = & \set{ A(r) \To_{\INTS} C(r) \ominus \non l |\ r\in R^{\INTS}[\non l,n]}\ \cup\nonumber\\
		      			 & \set{ A(r) \To_{\INTS} C(r) ! l |\ r\in R^{\INTS}[l,n]}.
		\end{align}
\end{itemize}
Then $D\equiv D'$.
\end{lem}

\begin{proof} The demonstration follows the inductive base and inductive step
of Lemma~\ref{lem:+PartialBox} where we consider the particular case
$\Box=\BEL$. Since here operations to obtain $D'$ modify only the consequent
of rules, verifying conditions when a given rule is applicable/discarded
reduces to clauses (2.1.2) and (4.1.2) of
Definitions~\ref{def:Applicability}--\ref{def:Discardability}, while
conditions for a rule being body-applicable/discarded are trivially treated.
Moreover, the analysis is narrowed to modalities $\OBL$, $\INT$, and $\INTS$
since rules for the other modalities are not affected by the operations of
the lemma. Finally, notice that the operations of the lemma do not erase
rules from $R$ to $R'$ but it may be the case that, given a rule $r$, if
removal or truncation operate on an element $c_{k}$ in $C(r)$, then $r\in
R[l]$ while $r\notin R'[l]$ for a given literal $l$ (removal of $l$ or
truncation at $c_{k}$).

\paragraph{$P(1)=+\partial_{X} q$, with $X\in \set{\OBL, \INT, \INTS}$.}{
We start by considering condition (2.2) of Definition~\ref{def:proofCond+X}
where a rule $r\in R^{X}[q,i]$ is applicable in $D$ at $i=1$ and $P(1)$. In
both cases when $q=l$ or $q\neq l$, $q$ is the first element of $C(r)$
since either we truncate chains at $l$, or we remove $\non l$ from them.
Therefore, $r$ is applicable in $D'$ as well. In the other direction, if $r$
is applicable in $D'$ at 1 and $P(1)$, then $r\in R$ has either $q$ as the
first element, or only $\non l$ precedes $q$. In the first case $r$ is
trivially applicable, while in the second case the applicability of $r$
follows from the hypothesis that $D\vdash +\partial l$ and $D\vdash
+\partial_{\OBL} \non l$ if $r\in R^{\OBL}$, or $D\vdash +\partial l$ and
$D\vdash -\partial_{\OBL} \non l$ if $r\in R^{\INTS}$.

Concerning condition (2.3.1) of Definition~\ref{def:proofCond+X} there is no
such rule $s$ in $R$, hence $s$ cannot be in $R'$ (we recall that at $P(1)$,
$s$ cannot be discarded in $D$ because that would imply a previous step in the
proof). Regarding the other direction, we have to consider the situation where
there is a rule $s\in R^{Y}[\non q, j]$ which is not in $R'^{Y}[\non q]$. This
is the case when the truncation has operated on $s\in R^{Y}[\non q, j]$ since
$l$ preceded $\non q$ in $C(s)$, making $s$ discarded in $D$ as well (either
when (i) $Y=\OBL$ or $Y=\INT$, or (ii) $D\vdash -\partial_{\OBL} \non l$ and
$Y=\INTS$). 

For (2.3.2) the reasoning is the same of the equivalent case in
Lemma~\ref{lem:+PartialBox} with the additional condition that rule $t$ may be
applicable in $D'$ at $P(1)$ but $q$ appears at index 2 in $C(t)$ in $D$.}

\paragraph{$P(n+1)=+\partial_{X} q$, with $X\in \set{\OBL, \INT, \INTS}$.}{
Again, let us suppose $r\in R[q,i]$ to be applicable in $D$ for $X$ at $i$ and
$P(n+1)$. By hypothesis and clauses (2.1.2) or (4.1.2) of
Definition~\ref{def:Applicability}, we conclude that $c_{k} \neq l$ and $q\neq
\non l$ ($\Conflict(\BEL,\INT)$ and $\Conflict(\BEL,\INTS)$). Thus, $r$ is
applicable in $D'$ by inductive hypothesis. The other direction sees $r\in
R'[q,i]$ applicable in $D'$ and either $\non l$ preceded $q$ in $C(r)$ in $D$,
or not. Since in the first case, the corresponding operation of the lemma is
the removal of $\non l$ from $C(r)$, while in the latter case no operations on
the consequent are done, the applicability of $r$ in $D$ at $P(n+1)$ is
straightforward.


For condition (2.3.1), the only difference between the inductive base is when
there is a rule $s$ in $R^Y[\non q,j]$ but $s\notin R'^Y[\non q,k]$. This
means that $l$ precedes $\non q$ in $C(s)$ in $D$, and thus by hypothesis $s$
is discarded in $D$. Notice that if $q=l$, then $R'^Y[\non l,k]=\emptyset$ for
any $k$ by the removal operation of the lemma, and thus condition (2.3.1) is
vacuously true.}

\paragraph{$P(1)=-\partial_{X} q$ and $P(n+1)=-\partial_{X} q$, with $X\in \MOD$.}{They trivially follow from the inductive base and inductive step.}
\end{proof}

\begin{lem}\label{lem:-PartialBEL}
	Let $D=(\FACTS,R,>)$ be a defeasible theory such that $D\vdash -\partial l$ and $D'=(\FACTS,R',>)$  be the theory obtained from $D$ where
\begin{align*}
	R'^{\INT} = & \set{ A(r) \To_{\INT} C(r)! \non l |\ r\in R^{\INT}[\non l,n]}.
\end{align*}
Moreover, 
\begin{itemize}
	\item if $D \vdash +\partial_{\OBL} l$, then
		\begin{align*}
			R'^{\OBL} = & \set{ A(r) \To_{\OBL} C(r) \ominus l |\ r\in R^{\OBL}[ l,n]};
		\end{align*}
	\item if $D \vdash -\partial_{\OBL} l$, then
		\begin{align*}
			R'^{\INTS} = & \set{ A(r) \To_{\INTS} C(r) ! \non l |\ r\in R^{\INTS}[\non l,n]}.
		\end{align*}
\end{itemize}
Then $D\equiv D'$.
\end{lem}
\begin{proof}
The demonstration is a mere variant of that of Lemma~\ref{lem:+PartialBEL}
since: (i) Proposition~\ref{prop:+PartialTHEN-Partial} states that
$+\partial_{X} m$ implies $-\partial_{X} \non m$ (mode $\DES$ is not
involved), and (ii) operations of the lemma are a subset of those of
Lemma~\ref{lem:+PartialBEL} where we switch $l$ with $\non l$, and the other
way around.
\end{proof}

\begin{lem}\label{lem:+PartialObl}
	Let $D=(\FACTS,R,>)$ be a defeasible theory such that $D\vdash +\partial_{\OBL} l$ and $D'=(\FACTS,R',>)$  be the theory obtained from $D$ where
\begin{align}\addtocounter{equation}{-3}
	R'^{\OBL} = & \set{ A(r) \To_{\OBL} C(r) ! \non l \ominus \non l |\ r\in R^{\OBL}[\non l,n]}\\
	R'^{\INTS} = & \set{ A(r) \To_{\INTS} C(r) \ominus \non l |\ r\in R^{\INTS}[\non l,n]}.
\end{align}
Moreover, 
\begin{itemize}
	\item if $D \vdash -\partial l$, then instead of (C1)
		\begin{align}\addtocounter{equation}{-2}
			R'^{\OBL} = & \set{ A(r) \To_{\OBL} C(r) ! \non l \ominus \non l |\ r\in R^{\OBL}[\non l,n]}\ \cup\nonumber \\
						& \set{ A(r) \To_{\OBL} C(r) \ominus l |\ r\in R^{\OBL}[l,n]};
		\end{align}
	\item if $D \vdash -\partial \non l$, then instead of (C2)
		\begin{align}
			R'^{\INTS} = & \set{ A(r) \To_{\INTS} C(r) \ominus \non l |\ r\in R^{\INTS}[\non l,n]}\ \cup \nonumber\\
						 & \set{ A(r) \To_{\INTS} C(r) ! l |\ r\in R^{\INTS}[l,n]}.
		\end{align}
\end{itemize}
Then $D\equiv D'$.
\end{lem}
\begin{proof}
Again, the proof is a variant of that of Lemma~\ref{lem:+PartialBEL}
that differs only when truncation and removal operate on a consequent at the
same time.

A CTD is relevant whenever its elements are proved as obligations.
Consequently, if $D$ proves $\OBL l$, then $\OBL \non l$ cannot hold. If this
is the case, then $\OBL \non l $ cannot be violated and elements following
$\non l$ in obligation rules cannot be triggered. Nonetheless, the inductive
base and inductive step do not significantly differ from those of
Lemma~\ref{lem:+PartialBEL}. In fact, even operation (1) involving truncation
and removal of $\non l$ does not affect the equivalence of conditions for
being applicable/discarded between $D$ and $D'$.
\end{proof}

Proofs for Lemmas~\ref{lem:-PartialObl}--\ref{lem:-PartialGIIS} are
not reported. As stated for Lemma~\ref{lem:+PartialObl}, they are
variants of that for Lemma~\ref{lem:+PartialBEL} where the modifications concern the set of rules on which we operate. The underlying motivation is
that truncation and removal operations affect when a rule is
applicable/discarded as shown before where we have proved that, given a rule
$s$ and a literal $\non q$, it may be the case that $\non q \notin C(s)$ in
$R'$ while the opposite holds in $R$. Such modifications reflect
only the nature of the operations of truncation and removal while they do not depend on the mode of the rule involved.

\begin{lem}\label{lem:-PartialObl}
	Let $D=(\FACTS,R,>)$ be a defeasible theory such that $D\vdash -\partial_{\OBL} l$ and $D'=(\FACTS,R',>)$ be the theory obtained from $D$ where
\begin{align*}
	R'^{\OBL} = & \set{ A(r) \To_{\OBL} C(r) ! l \ominus l |\ r\in R^{\OBL}[l,n]}.
\end{align*}
Moreover, 
\begin{itemize}
	\item if $D \vdash -\partial l$, then
		\begin{align*}
			R'^{\INTS} = & \set{ A(r) \To_{\INTS} C(r) ! \non l |\ r\in R^{\INTS}[\non l,n]}.
		\end{align*}
\end{itemize}
Then $D\equiv D'$.
\end{lem}

\begin{lem}\label{lem:+PartialDes}
	Let $D=(\FACTS,R,>)$ be a defeasible theory such that $D\vdash +\partial_{\DES} l$, $D\vdash +\partial_{\DES} \non l$, and $D'=(\FACTS,R',>)$ be the theory obtained from $D$ where
\begin{align*}
	R'^{\GOAL} = & \set{ A(r) \To_{\GOAL} C(r) ! l \ominus l |\ r\in R^{\GOAL}[l,n]}\ \cup\\
				 & \set{ A(r) \To_{\GOAL} C(r) ! \non l \ominus \non l |\ r\in R^{\GOAL}[\non l,n]}.
\end{align*}
Then $D\equiv D'$.
\end{lem}

\begin{lem}\label{lem:-PartialDes}
	Let $D=(\FACTS,R,>)$ be a defeasible theory such that $D\vdash -\partial_{\DES} l$ and $D'=(\FACTS,R',>)$ be the theory obtained from $D$ where
\begin{align*}
	R'^{\DES} = & \set{ A(r) \To_{\DES} C(r) \ominus l |\ r\in R^{\DES}[l,n]}\\
	R'^{\GOAL} = & \set{ A(r) \To_{\GOAL} C(r) \ominus l |\ r\in R^{\GOAL}[l,n]}.
\end{align*}
Then $D\equiv D'$.
\end{lem}

\begin{lem}\label{lem:+PartialGIIS}
	Let $D=(\FACTS,R,>)$ be a defeasible theory such that $D\vdash +\partial_{X} l$, with $X\in \set{\GOAL, \INT, \INTS}$, and $D'=(\FACTS,R',>)$ be the theory obtained from $D$ where
\begin{align*}
	R'^{X} = & \set{ A(r) \To_{X} C(r) ! l |\ r\in R^{X}[l,n]}\ \cup\\
	 		 & \set{ A(r) \To_{X} C(r) \ominus \non l |\ r\in R^{X}[\non l,n]}.
\end{align*}
Then $D\equiv D'$.
\end{lem}

\begin{lem}\label{lem:-PartialGIIS}
	Let $D=(\FACTS,R,>)$ be a defeasible theory such that $D\vdash -\partial_{X} l$, with $X\in \set{\GOAL, \INT, \INTS}$, and $D'=(\FACTS,R',>)$ be the theory obtained from $D$ where
\begin{align*}
	R'^{X} = & \set{ A(r) \To_{X} C(r) \ominus l |\ r\in R^{X}[l,n]}.
\end{align*}
Then $D\equiv D'$.
\end{lem}

\begin{lem}\label{lem:Prove}
	Let $D=(\FACTS,R,>)$ be a defeasible theory and $l\in \LIT$ such that (i) $Xl \notin \FACTS$, (ii) $\neg X l\not\in \FACTS$ and $Y \non l\not\in\FACTS$ with $Y=X$ or $\Aconf{Y}{X}$, (iii) $\exists r\in
R^{X}[l,1] \cup R^{\BEL, X}[l,1]$, (iv) $A(r)=\emptyset$, and (v)
$R^{X}[\non l] \cup R^{\BEL,X}[\non l] \cup R^{Y}[\non l]\setminus
R_{infd}\subseteq r_{inf}$, with $X \in \MOD \setminus \set{\DES}$. Then
$D\vdash+\partial_{X} l$.
\end{lem}
\begin{proof}	
To prove $Xl$, Definition~\ref{def:proofCond+X} must be taken into
consideration: since hypothesis (i) falsifies clause (1), then clause (2)
must be the case. Let $r$ be a rule that meets the conditions of the lemma.
Hypotheses (iii) and (iv) state that $r$ is applicable for $X$. In
particular, if $r=s^{\Diamond}\in R^{\BEL,X}$ then $s$ is Conv-applicable.
Finally, for clause (2.3) we have that all rules for $\non l$ are inferiorly
defeated by an appropriate rule with empty antecedent for $l$, but a rule with
empty body is applicable. Consequently, all clauses for proving
$+\partial_{X}$ are satisfied. Thus, $D\vdash+\partial_{X} l$.
\end{proof}

\begin{lem}\label{lem:ProveDES}
	Let $D=(\FACTS,R,>)$ be a defeasible theory and $l\in \LIT$ such that (i) $\DES l \not\in \FACTS$, (ii) $\neg\DES l \not\in \FACTS$, (iii) $\exists r\in R^{\DES}[l,1] \cup R^{\BEL, \DES}[l,1]$, (iv) $A(r)=\emptyset$, and (v) $r_{sup}=\emptyset$. Then $D\vdash+\partial_{\DES} l$.
\end{lem}
\begin{proof}	
The demonstration is analogous to that for Lemma~\ref{lem:Prove} since all lemma's hypotheses meet clause (2) of Definition~\ref{def:proofCond+DES}.
\end{proof}

\begin{lem}\label{lem:NORules}
	Let $D=(\FACTS,R,>)$ be a defeasible theory and $l\in \LIT$ such that $l,\, Xl\not\in \FACTS$ and $R^{X}[l] \cup R^{\BEL,X}[l] =\emptyset$, with $X \in \MOD$. Then $D\vdash-\partial_X l$.
\end{lem}
\begin{proof}
Conditions (1) and (2.2) of Definitions~\ref{def:proofCond-DES} and \ref{def:proofCond-X} are vacuously satisfied with the same comment for $R^{\BEL,X}$ in Lemma~\ref{lem:Prove}.
\end{proof}

\begin{lem}\label{lem:Refute}
	Let $D=(\FACTS,R,>)$ be a defeasible theory and $l\in \LIT$ such that (i) $X\non l \not\in \FACTS$, (ii) $\neg X \non l\not\in \FACTS$ and $Y l\not\in\FACTS$ with $Y=X$ or $\Aconf{Y}{X}$, (iii) $\exists r\in R^{X}[l,1] \cup R^{\BEL,X}[l,1]$, (iv) $A(r)=\emptyset$, and (v) $r_{sup}=\emptyset$, with $X \in \MOD$. Then $D\vdash-\partial_{X}\non l$.
\end{lem}
\begin{proof}
Let $r$ be a rule in a theory $D$ for which the conditions of the lemma hold.
It is easy to verify that clauses (1) and (2.3) of
Definitions~\ref{def:proofCond-DES} and \ref{def:proofCond-X} are satisfied
for $\non l$.
\end{proof}

\Termination*
\begin{proof}
 Every time Algorithms~\ref{alg:proved}~\textsc{Proved} or
 \ref{alg:refuted}~\textsc{Refuted} are invoked, they both modify a subset of the set of rules $R$, which is finite by hypothesis. Consequently, we have their termination. Moreover, since $|R| \in O(S)$ and each rule can be accessed in constant time, we obtain that their computational complexity is $O(S)$.
\end{proof}

\Complexity*
\begin{proof}
 The most important part to analyse concerning termination of
Algorithm~\ref{alg:defeasible} \textsc{DefeasibleExtension} is the
\textbf{repeat/until} cycle at lines \ref{CDrigaRepeat}--\ref{CDrigaUntil}.
Once an instance of the cycle has been performed, we are in one of the
following, mutually exclusive situations:

\begin{enumerate}

\item No modification of the extension has occurred. In this case, line \ref{CDrigaUntil} ensures the termination of the algorithm;

\item The theory has been modified with respect to a literal in $HB$. Notice
that the algorithm takes care of removing the literal from $HB$ once the
suitable operations have been performed (specifically, at line
\ref{MMrigaUpHB} of Algorithm~\ref{alg:proved}~\textsc{Proved} and
\ref{alg:refuted}~\textsc{Refuted}). Since this set is finite, the
process described above eventually empties $HB$ and, at the next iteration of
the cycle, the extension of the theory cannot be modified. In this
case, the algorithm ends its execution as well.
\end{enumerate}

Moreover, Lemma~\ref{lem:ComplexityMMDisc} proved the termination of its
internal sub-routines.

In order to analyse complexity of the algorithm, it is of the utmost
importance to correctly comprehend Definition~\ref{def:TheorySize}. Remember
that the size of a theory is the number of \emph{all occurrences} of each
literal in every rule plus the number of the rules. The first term is
usually (much) bigger than the latter. Let us examine a theory with $x$
literals and whose size is $S$, and consider the scenario when an algorithm
$A$, looping over all $x$ literals of the theory, invokes an inner procedure $P$
which selectively deletes a literal given as input from all the rules of the
theory (no matter to what end). A rough computational complexity would be
$O(S^{2})$, given that, when one of the $x \in O(S)$ literal is selected, $P$
removes all its occurrences from every rule, again $O(S)$.

However, a more fined-grained analysis shows that the complexity of $A$ is
lower. The mistake being to consider the complexity of $P$ separately from the
complexity of the external loop, while instead they are strictly dependent.
Indeed, the overall number of operations made by the sum of all loop
iterations cannot outrun the number of occurrences of the literals, $O(S)$,
because the operations in the inner procedure directly decrease, iteration
after iteration, the number of the remaining repetitions of the outmost loop,
and the other way around. Therefore, the overall complexity is not bound by
$O(S)\cdot O(S) = O(S^{2})$, but by $O(S) + O(S) = O(S)$.

We can now contextualise the above reasoning to
Algorithm~\ref{alg:defeasible}~\textsc{DefeasibleExtension}, where $D$ is the
theory with size $S$. The initialisation steps (lines
\ref{CDriga1}--\ref{CDrigaSup1} and \ref{CDrigaUp}--\ref{CDrigaPulisci}) add
an $O(S)$ factor to the overall complexity. The main cycle at lines
\ref{CDrigaRepeat}--\ref{CDrigaUntil} is iterated over $HB$, whose cardinality
is in $O(S)$. The analysis of the preceding paragraph implies that invoking
Algorithm~\ref{alg:proved}~\textsc{Proved} at lines \ref{CDrigaIfMM} and
\ref{CDrigaModRepeat} as well as invoking
Algorithm~\ref{alg:refuted}~\textsc{Refuted} at lines \ref{CDrigaIfDisc},
\ref{CDrigaIfDisc2}, \ref{CDrigaRepeatDisc1} and \ref{CDrigaRepeatDisc2}
represent an additive factor $O(S)$ to the complexity of \textbf{repeat/until}
loop and \textbf{for} cycle at lines \ref{CDrigaFor2}--\ref{CDrigaEndFor2} as
well. Finally, all operations on the set of rules and the superiority relation
require constant time, given the implementation of data structures proposed.
Therefore, we can state that the complexity of the algorithm is $O(S)$.
\end{proof}

\Correctness*
\begin{proof}
As already argued at the beginning of the section, the aim of
Algorithm~\ref{alg:defeasible} \mbox{\textsc{DefeasibleExtension}} is to
compute the defeasible extension of a given defeasible theory $D$ through
successive transformations on the set of facts and rules, and on the
superiority relation: at each step, they compute a simpler theory while
retaining the same extension. Again, we remark that the word `simpler' is used
to denote a theory with fewer elements in it. Since we have already proved the
termination of the algorithm, it eventually comes to a fixed-point theory where
no more operations can be made.

In order to demonstrate the soundness of Algorithm~\ref{alg:defeasible}
\mbox{\textsc{DefeasibleExtension}}, we show in the list below that all the
operations performed by the algorithm are justified by Proposition \ref{prop:+PartialTHEN-Partial} and described in
Lemmas~\ref{lem:addRandSup}--\ref{lem:Refute}, where
we prove the soundness of each operation involved.

\begin{enumerate}
 \item Algorithm~\ref{alg:defeasible} \textsc{DefeasibleExtension}:
 \begin{itemize}
  \item Lines \ref{CDrigaOUTCopy}--\ref{CDrigaR-} and \ref{CDrigaSup1}: Lemma \ref{lem:addRandSup};
  \item Line \ref{CDrigaIfMM}: item \ref{item:proved}. below;
  \item Line \ref{CDrigaIfDisc}: item \ref{item:refuted}. below;
  \item Line \ref{CDrigaIfDisc2}: Lemma~\ref{lem:NORules} and item \ref{item:refuted}. below;
  \item Line \ref{CDrigaProvedDES}: Lemma~\ref{lem:ProveDES} and item \ref{item:proved}. below;
  \item Lines \ref{CDrigaRepeatDisc1}--\ref{CDrigaRepeatDisc2}: Lemma~\ref{lem:Refute} and item \ref{item:refuted}. below;
  \item Line \ref{CDrigaModRepeat}: Lemma~\ref{lem:Prove} and item \ref{item:proved}. below;
 \end{itemize}
 \item \label{item:proved} Algorithm~\ref{alg:proved} \textsc{Proved}:
 \begin{itemize}
  \item Line \ref{MMrigaIfD}: Lemma~\ref{lem:Refute} and item \ref{item:refuted}. below;
  \item Line \ref{MMrigaIfBELpINTnon-p}: Part 2. of Proposition \ref{prop:+PartialTHEN-Partial} and item \ref{item:refuted}. below;
  \item Line \ref{MMrigaIfBEL-OBLpINTSnon-p}: Part 3. of Proposition \ref{prop:+PartialTHEN-Partial} and item \ref{item:refuted}. below;
  \item Lines \ref{MMrigaR}--\ref{MMrigaSup}: Lemma~\ref{lem:+PartialBox};
  \item \textsc{Case $\BEL$} at lines \ref{MMrigaCaseB}--\ref{MMrigaBIf-OR}: Lemma~\ref{lem:+PartialBEL};
  \item \textsc{Case $\OBL$} at lines \ref{MMrigaCaseO}--\ref{MMrigaOIfR2}: Lemma~\ref{lem:+PartialObl};
  \item \textsc{Case $\DES$} at lines \ref{MMrigaCaseD}--\ref{MMrigaEndDIf}: Lemma~\ref{lem:+PartialDes};
  \item \textsc{Otherwise} at lines \ref{MMrigaOtherwise}--\ref{MMrigaOtherR2}: Lemma~\ref{lem:+PartialGIIS};
 \end{itemize}
 \item \label{item:refuted} Algorithm~\ref{alg:refuted} \textsc{Refuted}:
 \begin{itemize}
  \item Lines \ref{DrigaCaseUpR2}--\ref{DrigaCaseUpSup}: Lemma~\ref{lem:-PartialBox};
  \item \textsc{Case $\BEL$} at lines \ref{DrigaCaseB}--\ref{DrigaCaseBIf-OR}: Lemma~\ref{lem:-PartialBEL};
  \item \textsc{Case $\OBL$} at lines \ref{DrigaCaseO}--\ref{DrigaCaseOIfR}: Lemma~\ref{lem:-PartialObl};
  \item \textsc{Case $\DES$} at lines \ref{DrigaCaseD}--\ref{DrigaCaseDR}: Lemma~\ref{lem:-PartialDes};
  \item \textsc{Otherwise} at lines \ref{DrigaOtherwise}--\ref{DrigaCaseOtherR}: Lemma~\ref{lem:-PartialGIIS};
 \end{itemize}
\end{enumerate}

\noindent The result of these lemmas is that whether a literal is defeasibly
proved or not in the initial theory, so it will be in
the final theory. This proves the soundness of the algorithm.

Moreover, since (i) all lemmas show the equivalence of the two theories, and
(ii) the equivalence relation is a bijection, this also demonstrates the
completeness of Algorithm~\ref{alg:defeasible} \textsc{DefeasibleExtension}.
\end{proof}

\end{document}